\documentclass[moor]{informs-rad}  

\usepackage{natbib}
 \NatBibNumeric
 \def\bibfont{\small}%
 \def\bibsep{\smallskipamount}%
 \bibpunct[, ]{[}{]}{,}{n}{}{,}%

\usepackage[colorlinks=true,breaklinks=true,bookmarks=true,urlcolor=blue,
     citecolor=blue,linkcolor=blue,bookmarksopen=false,draft=false]{hyperref}
\newcommand{\revcolor}[1]{{\color{black}{#1}}}

\def\EMAIL#1{\href{mailto:#1}{#1}}

\TheoremsNumberedThrough     

\EquationsNumberedThrough    

\usepackage{varwidth}
\usepackage{amsmath}
\usepackage{xfrac}
\usepackage{subfigure}
\usepackage{csquotes}
\usepackage{comment}
\usepackage[table]{xcolor}
\usepackage{bbm}
\usepackage{hhline}
\usepackage{mathtools}
\usepackage{booktabs}
\usepackage{ifthen}
\usepackage{algpseudocode,algorithm,algorithmicx}

\algrenewcommand\algorithmicrequire{\textbf{Precondition:}}
\algrenewcommand\algorithmicensure{\textbf{Postcondition:}}
\usepackage{mathrsfs}
\usepackage{dsfont}
\definecolor{maroon}{rgb}{0.52, 0, 0}
\definecolor{dgreen}{rgb}{0.0, 0.5, 0.0}
\definecolor{ballblue}{rgb}{0.13, 0.67, 0.8}
\definecolor{royalblue(web)}{rgb}{0.25, 0.41, 0.88}
\definecolor{bleudefrance}{rgb}{0.19, 0.55, 0.91}
\definecolor{royalazure}{rgb}{0.0, 0.22, 0.66}
\usepackage{hyperref}
\hypersetup{
	colorlinks = true,
	linkcolor=maroon,
	citecolor=maroon,
	linkbordercolor = {white}
}
\usepackage{cleveref}
\usepackage{enumitem}
\usepackage{multirow}

\usepackage{tikz}
\usepackage{pgfplots}
\usetikzlibrary{patterns}
\usepgfplotslibrary{fillbetween}
\usetikzlibrary{intersections}
 \usepackage{natbib}
\usetikzlibrary{arrows, decorations.markings}

\tikzstyle{vecArrow} = [thick, decoration={markings,mark=at position
	1 with {\arrow[semithick]{open triangle 60}}},
	double distance=1.4pt, shorten >= 5.5pt,
	preaction = {decorate},
postaction = {draw,line width=1.4pt, white,shorten >= 4.5pt}]
\tikzstyle{innerWhite} = [semithick, white,line width=1.4pt, shorten >= 4.5pt]

	\usepackage{accents}
\providecommand{\given}{}
\DeclarePairedDelimiterX{\card}[1]{\lvert}{\rvert}{\renewcommand\given{\nonscript\:\delimsize\vert\nonscript\:\mathopen{}}#1}
\DeclarePairedDelimiterX{\abs}[1]{\lvert}{\rvert}{\renewcommand\given{\nonscript\:\delimsize\vert\nonscript\:\mathopen{}}#1}
\DeclarePairedDelimiterX{\norm}[1]{\lVert}{\rVert}{\renewcommand\given{\nonscript\:\delimsize\vert\nonscript\:\mathopen{}}#1}
\DeclarePairedDelimiterX{\tuple}[1]{\lparen}{\rparen}{\renewcommand\given{\nonscript\:\delimsize\vert\nonscript\:\mathopen{}}#1}
\DeclarePairedDelimiterX{\parens}[1]{\lparen}{\rparen}{\renewcommand\given{\nonscript\:\delimsize\vert\nonscript\:\mathopen{}}#1}
\DeclarePairedDelimiterX{\brackets}[1]{\lbrack}{\rbrack}{\renewcommand\given{\nonscript\:\delimsize\vert\nonscript\:\mathopen{}}#1}
\DeclarePairedDelimiterX{\set}[1]\{\}{\renewcommand\given{\nonscript\:\delimsize\vert\nonscript\:\mathopen{}}#1}
\let\Pr\relax
\DeclarePairedDelimiterXPP{\Pr}[1]{\mathbb{P}}[]{}{\renewcommand\given{\nonscript\:\delimsize\vert\nonscript\:\mathopen{}}#1}
\DeclarePairedDelimiterXPP{\PrX}[2]{\mathbb{P}_{#1}}[]{}{\renewcommand\given{\nonscript\:\delimsize\vert\nonscript\:\mathopen{}}#2}
\DeclarePairedDelimiterXPP{\Ex}[1]{\mathbb{E}}[]{}{\renewcommand\given{\nonscript\:\delimsize\vert\nonscript\:\mathopen{}}#1}
\DeclarePairedDelimiterXPP{\ExX}[2]{\mathbb{E}_{#1}}[]{}{\renewcommand\given{\nonscript\:\delimsize\vert\nonscript\:\mathopen{}}#2}

\DeclarePairedDelimiterXPP{\1}[1]{\mathds{1}}[]{}{\renewcommand\given{\nonscript\:\delimsize\vert\nonscript\:\mathopen{}}#1}

\newcommand*{\eps}{\epsilon}

\newcommand{\X}{\mathcal{X}}
\newcommand{\Y}{\mathcal{Y}}
\newcommand{\Exv}[2]{\mathbb{E}_{#1}\left[#2\right]}
\newcommand{\V}{\mathcal{V}}
\newcommand{\T}{\tau}

\newcommand{\Z}{\mathcal{Z}}
\newcommand{\Sm}{\mathscr{S}}
\newcommand{\La}{\mathscr{L}}
\newcommand{\Po}{\mathcal{P}}

\newcommand{\opton}{\texttt{OPT-ONLINE}}
\newcommand{\lpon}{\texttt{LP-ONLINE}}
\newcommand{\x}[1]{\mathcal{X}_{#1}}

\newcommand{\Fi}[1]{{F}_{#1}}

\begin{document}


\RUNAUTHOR{Anari et al.}

\RUNTITLE{Linear Programming Based Near-Optimal Pricing for Laminar Bayesian Online Selection}

\TITLE{Linear Programming Based Near-Optimal Pricing for Laminar Bayesian Online Selection}

\ARTICLEAUTHORS{%
\AUTHOR{Nima Anari}
\AFF{Computer Science Department, Stanford University, Stanford, CA, \EMAIL{anari@cs.stanford.edu}}
\AUTHOR{Rad Niazadeh}
\AFF{Booth School of Business, University of Chicago, Chicago, IL, \EMAIL{rad.niazadeh@chicagobooth.edu}}
\AUTHOR{Amin Saberi}
\AFF{Management Science and Engineering, Stanford University, Stanford, CA, \EMAIL{saberi@stanford.edu}}
\AUTHOR{Ali Shameli}
\AFF{Instacart, San Francisco, CA, \EMAIL{ali.shameli@gmail.com }}
} 

\ABSTRACT{%
The Bayesian online selection problem aims to design a pricing scheme for a sequence of arriving buyers that maximizes the expected social welfare (or revenue) subject to different structural constraints. Inspired by applications with a hierarchy of service, this paper focuses on the cases where a laminar matroid characterizes the set of served buyers. We give the first Polynomial-Time Approximation Scheme (PTAS) for the problem when the laminar matroid has constant depth. Our approach is based on rounding the solution of a hierarchy of linear programming relaxations that approximate the optimum online solution with any degree of accuracy, plus a concentration argument showing that rounding incurs a small loss. We also study another variation, which we call the production-constrained problem. The allowable set of served buyers is characterized by a collection of production and shipping constraints that form a particular example of a laminar matroid. \revcolor{Using a similar LP-based approach, we design a PTAS for this problem, although in this special case the depth of the underlying laminar matroid is not necessarily a constant.} The analysis exploits the negative dependency of the optimum selection rule in the lower levels of the laminar family. Finally,  to demonstrate the generality of our technique, we employ the linear programming-based approach employed in the paper to re-derive some of the classic prophet inequalities known in the literature --- as a side result.
}%



\maketitle

\section{Introduction}

This paper revisits a canonical problem in algorithm design: how should a planner allocate a limited number of goods or resources to a set of agents arriving over time? Examples of this  problem range from selling seats in a concert hall to online retail and sponsored-search auctions. In many of these applications, it is often reasonable to assume that each agent has a private valuation drawn from a known distribution. Moreover, the allocation is often subject to combinatorial constraints such as matroids, matchings, or knapsacks. The goal of the planner is  to maximize social-welfare, i.e. the total value of served agents.\footnote{In a single-parameter Bayesian setting like in this paper, the problem of maximizing the revenue can also be reduced to the maximization of welfare with a simple transformation using (ironed) virtual values~\citep{myerson1981optimal}.} This problem, termed as \emph{Bayesian online selection}, originated from the seminal work of \cite{krengel1978semi} and has since been studied quite extensively in probability theory, operations research, and computer science (see \cite{lucier2017economic} for a comprehensive survey).

A common approach to the above stochastic online optimization problem is to obtain \emph{``prophet inequalities''} which  evaluate the performance of an online algorithm relative to an offline ``omniscient  prophet'', who knows the valuation of each agent and therefore can easily maximize the social-welfare.  The upshot of a significant line of work studying prophet inequalities is that in many complex combinatorial settings there exist simple and elegant take-it-or-leave-it pricing rules that obtain a constant factor approximation with respect to the omniscient prophet benchmark. Examples include but are not limited to single-item sale~\citep{samuel1984comparison, hill1982comparisons, correa2017posted}, matroids~\citep{hajiaghayi2007automated,chawla2010multi,kleinberg2012matroid}, matchings~\citep{chawla2010multi, alaei2012online,alaei2014bayesian,gravin2019prophet}, intersections of matroids~\citep{feldman2016online,ma2019algorithms}, and even combinatorial auctions~\citep{feldman2013simultaneous}. Somewhat surprisingly, it is also often possible to prove matching information theoretic lower-bounds e.g. for matroids~\citep{niazadeh2018prophet}.


In this paper, we deviate from the above framework and dig into the question of characterizing and computing {\em optimum online policies}. Given the sequence of value distributions,   Bellman's ``principle of optimality''~\citep{bellman1954theory} proposes a simple dynamic programming that computes the optimum online policy for all of the above problems. Unfortunately, the dynamic program needs to track the full state of the system and therefore it often requires exponential time and space.



While there are fairly strong lower bounds for the closely related computation of Markov Decision Processes (see  \cite{papadimitriou1987complexity} for the PSPACE-hardness of the general Markov decision processes with partial observations), the computational complexity of the stochastic online optimization problems with a concise combinatorial structure, like the one we are considering here, is poorly understood.  \revcolor{Notably,  since the appearance of an early conference version of our paper, the work of \cite{papadimitriou2021online} has established the PSPACE-hardness of Bayesian online matching problem, which is among very few results shedding light on the hardness of structured instances of stochastic online optimization. However, this result does not apply to the laminar matroid Bayesian online selection problem. To the best of our knowledge, there is no formal hardness result for this special class of stochastic online optimization, even for general matroid Bayesian online selection --- and hence establishing any computational complexity hardness for this problem is still open.}
Here,  we ask whether it is possible to approximate the optimum online in polynomial time, and obtain improved approximation factors compared to those derived from the prophet inequalities. If we answer this question in the affirmative,  it justifies the optimum online policy as a less pessimistic benchmark compared to the omniscient prophet benchmark.

\subsection{Our contribution}

We focus on two special cases of the Bayesian online selection problem.  First, we consider the problem of \emph{laminar Bayesian selection}, which is a special case of the well-known matroid Bayesian online selection problem studied in \cite{kleinberg2012matroid}, when the underlying matroid is laminar. In this problem, elements arrive over time with values drawn from heterogeneous but known independent distributions. Laminar matroids are a special case of matroids. A rooted directed tree whose leaves correspond to these elements and has a capacity on each of its internal nodes specifies the laminar matroid feasibility constraints as follows. For every internal node of the tree, a feasible set of elements (a.k.a. an independent set) does not contain more than the capacity of this internal node from the set of leaves that are connected to this node through a directed path in the tree. The depth of this rooted tree represents the depth of our laminar matroid. 

The above constraints can be seen as capturing the limited capacity of the firm in delivering products or services at different geographic levels. The constraint corresponding to the root captures the total capacity of the firm and the ones corresponding to the internal nodes correspond to the capacity of possibly state, region, city, or neighborhood. As a concrete example, suppose the service network of a firm, headquartered in San Francisco (\texttt{SFO}), is as in \Cref{fig:map}. The firm has some capacity at each city. Moreover, delivering a service at a node will take one unit of capacity from each node on the path connecting root to that node. For example, delivering a service at \texttt{BNA} requires a unit of capacity from \texttt{BNA}, \texttt{ORD} and \texttt{SFO}, while delivering a service at \texttt{SEA} only uses a unit of capacity from \texttt{SEA} and \texttt{SFO}. Under this hierarchical structure, service requests with known value distributions arrive at different nodes.


\begin{figure}[ht]
  \centering
  \includegraphics[width=0.8\columnwidth]{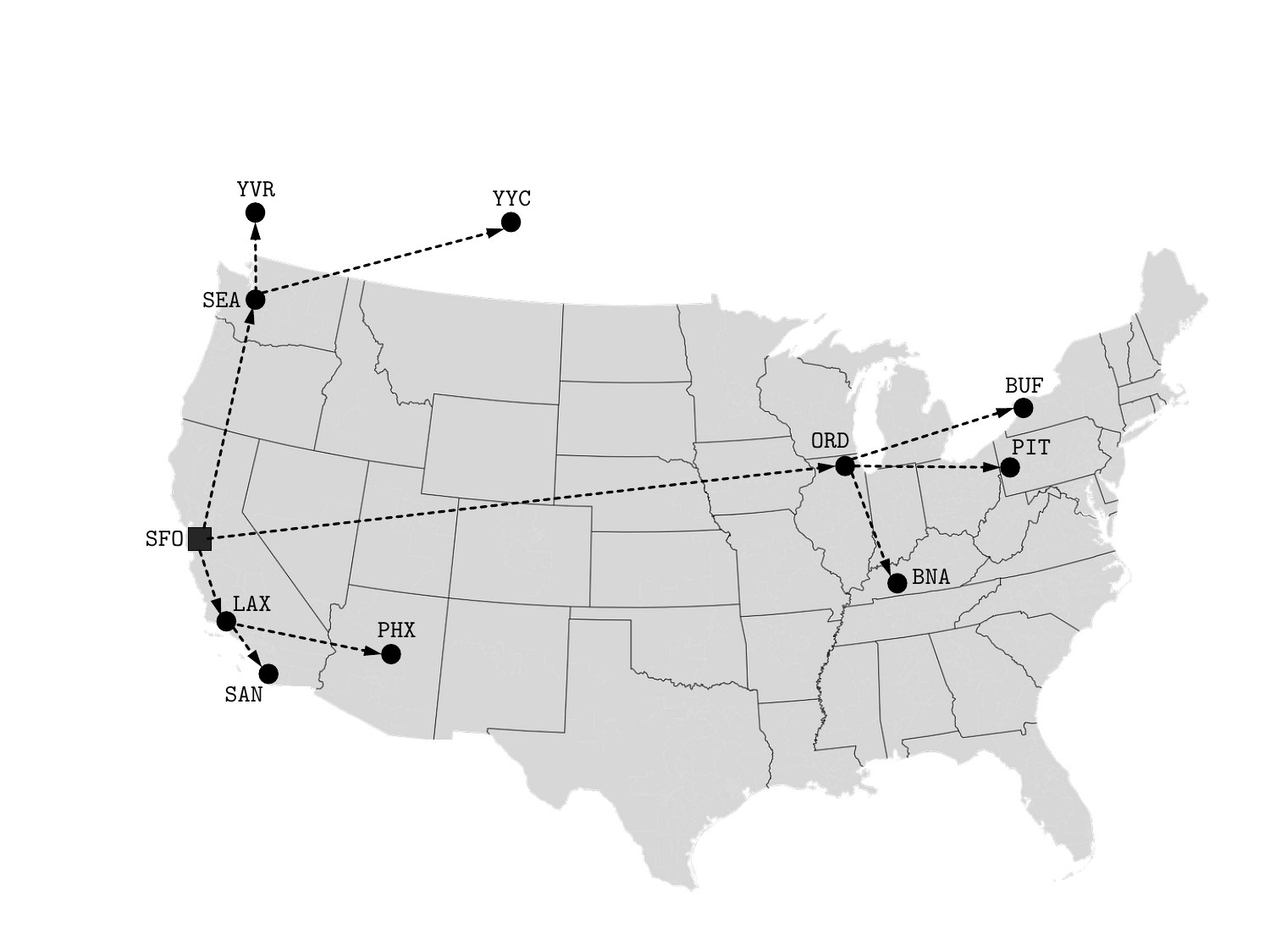}
     \caption{The ``hierarchical'' map of service level network.}
      \label{fig:map}
\end{figure}


We also consider another variation of the laminar Bayesian selection motivated by production and shipping constraints.  
Consider a firm producing  multiple copies of different product types over time. The firm offers the products to arriving unit-demand buyers who are interested in one type of product.  The goal is to maximize social welfare (or revenue) subject to two types of constraints. First, at any time, the total number of sold items of each type is no more than the number of produced items. Second, the total number of items sold does not exceed the total shipping capacity. We term this stochastic online optimization problem as \emph{production constrained Bayesian selection} .


We show that both of the above problems are amenable to polynomial-time approximations with any degrees of accuracy. This is done by introducing linear programming relaxations for these problems and then designing appropriate rounding schemes through pricing.
\vspace{3mm}
\begin{displayquote}
\emph{\textbf{Main Results.} We give Polynomial Time Approximation Schemes (PTAS) for the laminar Bayesian selection problem when the depth of the laminar family is bounded by a constant, as well as the production constrained Bayesian selection problem.}
\end{displayquote}
\vspace{3mm}
Finally, to further showcase the LP based approach employed in the paper, we use it to derive classic prophet inequality results known in the literature for the single-item Bayesian online selection problem~\citep{krengel1978semiamarts,hill1982comparisons,abolhassani2017beating,correa2017posted}. For the case of nonidentical distributions, we introduce a new adaptive pricing policy that obtains $\tfrac{1}{2}$  of the expected value of the prophet. For identical distributions we show that a simple single-price policy obtains $(1-\tfrac{1}{e})$ fraction of that benchmark. 



\subsection{Overview of the techniques} We start by characterizing the optimum online policy for both of the problems through a \emph{Linear Programming} formulation. The LP formulation captures Bellman's dynamic program by tracking the state of the system through allocation and state variables (see \cref{sec:production-simple} for more details) and express the conditions for a policy to be \emph{feasible}  and \emph{online implementable} as linear constraints. Our method for capturing optimum online policy through linear programming resembles the dynamic programming to linear programming conversion technique introduced in \cite{de2003linear}. The resulting LPs are exponentially big but they accept polynomial-sized \emph{relaxations} with a small error. Furthermore, the relaxations can be rounded and implemented as online implementable policies, in the same way as exponential-sized LPs.



More precisely, we propose a hierarchy of linear programming relaxations that systematically strengthen the commonly used ``expected'' LP formulation of  the problem and approximate the optimum solution with any degrees of accuracy. The first level of our LP hierarchy is the expected relaxation, which is a simple linear program requiring that the allocation satisfies the capacity constraint(s) only in expectation. \revcolor{It is well-known that the gap between this LP and the optimum online policy is 2~\citep{duetting2017prophet,alaei2014bayesian}.}  At the other extreme, the linear program is of exponential size and is equivalent to the dynamic program.


Given $\eps$ as the error parameter of the desired PTAS, we show how to choose a linear program that combines the constraints of these two LPs in a careful way to get $\eps$-close to the optimum solution. In a nutshell, this hierarchy is parametrized by how we divide up the capacity constraints into ``large'' and ``small''.  In the laminar Bayesian selection, we consider the tree corresponding to the laminar family of constraints. Our approach here is based on chopping the tree (with the constraints as its internal nodes) by a horizontal cut, and then marking the constraints above the cut as large and below the cut as small (left figure, \Cref{fig:intro}). The final relaxation then needs to respect all the small constraints exactly and all the large constraints only in expectation.

\begin{figure}[ht]
  \centering
  \includegraphics[width=1\columnwidth]{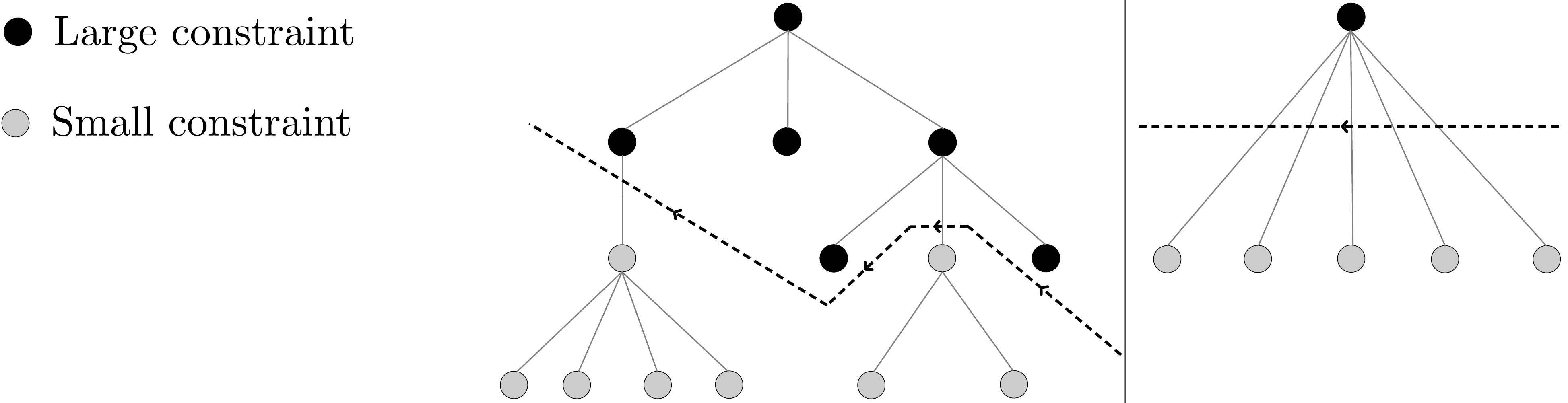}
     \caption{Characterization of our hierarchy of linear programming relaxations.}
     \label{fig:intro}
\end{figure}

Our final algorithms start by reducing the capacities of large bins by a factor of $(1-\eps)$ to create some slack, solve the corresponding LP relaxation, and then adaptively round the solution. A coupling argument shows that the LP solution can be implemented with an adaptive online pricing policy (potentially with randomized tie-breaking). However, the resulting online policy respects the large constraints only in expectation. 


In the production constrained Bayesian selection, we simply consider two cases based on shipping capacity being large or small (right figure, \Cref{fig:intro}). The main technical ingredient in the analysis for this algorithm is to establish a particular form of {\em negative dependency} between the allocation events of this policy. In fact, we show that the event that the optimum online policy makes an allocation decision at each time (for a certain type of buyer) is negatively dependent on the number of allocations made in the past (for the same type of buyer); this in turn leads to concentration results on the number of allocated items in large bins (e.g. see~\citep{dubhashi1998balls}), and shows that the policy only violates the large capacity constraints with a small probability.

The analysis of the above negative dependence uses a very careful argument that essentially establishes the submodularity of the value function of the dynamic program.  See \cref{sec:production-simple} for the details. Surprisingly, the negative dependence property of optimum online policies no longer holds for laminar matroids with arbitrary arrival order of elements. We present examples in which the event that more buyers accept the offered price leads the optimum online to offer a lower price to the next arriving buyer.  In this case, we use a different trick by carefully chopping the laminar tree and marking the constraints to ensure negative dependence. See \Cref{sec:laminar} for the marking algorithm and its analysis.

Finally, as a side result and to demonstrate an \emph{alternative application} of our LP framework, we reinvent some classic results in the prophet inequality literature using the LP technique we had earlier.  To this end, we focus on the classic single-item prophet inequality problem, where the ordering of buyers is unknown, but their values are independently drawn from known distributions. Inspired by the idea of hierarchy of linear programming relaxations used above, we consider two linear programs: expected relaxation, in which the number of sold items in expectation is at most one, and the optimum online LP when the ordering is known. We first solve the expected relaxation (which does not need to know the ordering of the buyers). Then, we show how to modify this optimal solution to be feasible in the optimum online LP. For the case of non-identical distributions, we introduce a modification that only loses $\tfrac{1}{2}$ of the expected objective value, and for the case of identical distributions we introduce a new modification that only loses $\tfrac{1}{e}$ fraction of the expected objective value. Interestingly, the resulting policies are order oblivious and are simple pricing policies (static for identical and adaptive for non-identical distributions) as mentioned earlier.

\subsection{Further related work.} Besides the combinatorial settings mentioned earlier, constraints such as knapsack~\citep{feldman2016online}, $k$-uniform matroids (for better bounds)~\citep{hajiaghayi2007automated,alaei2014bayesian}, or even general downward-closed~\citep{rubinstein2016beyond} have been studied in the literature on prophet inequalities. Moreover, many variations such as prophet inequalities with limited samples form the distributions or inaccurate priors~\citep{azar2014prophet,dutting2019posted,caramanis2021single,nuti2022secretary},  i.i.d. and random order prophets~\citep{hill1982comparisons,esfandiari2017prophet,abolhassani2017beating, correa2017posted, azar2018prophet,correa2019prophet,correa2021prophet}, and free-order prophets~\citep{yan2011mechanism,beyhaghi2018improved} have been explored, and connections to the price of anarchy~\citep{duetting2017prophet}, online contention resolution schemes~\citep{alaei2014bayesian,feldman2016online, lee2018optimal,jiang2022tight,fu2022oblivious}, and online combinatorial optimization~\citep{gobel2014online} have been of particular interest in this literature. Finally, techniques and results in this literature had an immense impact on mechanism design~\citep{chawla2010multi,cai2012optimal,feldman2013simultaneous,babaioff2015simple,cai2016duality,chawla2016mechanism}. For a full list, refer to \citep{lucier2017economic}.

Stochastic optimization problems with similar flavors, either online or offline, have also been massively studied both in the operations research and the computer science literature. Examples include (but not limited to) stochastic knapsack~\citep{dean2004approximating,bhalgat2011improved,ma2014improvements}, online stochastic matching~\citep{manshadi2012online,jaillet2014online,huang2021online}, online matching with stochastic rewards~\citep{goyal2022online,mehta2012online}, stochastic assortment optimization and pricing~\citep{goyal2016near,rusmevichientong2017dynamic, ma2018dynamic,feng2019linear}, stochastic probing~\citep{chen2009approximating,gupta2016algorithms}, and pre-planning in stochastic optimization~\citep{immorlica2004costs}.   There are also other papers that study computational questions related to prophet inequalities.  For example,  \cite{agrawal2020optimal} study the optimal ordering problem in free-order prophet inequalities and establishes its NP-hardness, and the work of \cite{fu2018ptas} obtains a PTAS for this problem.  The closest work in the operations research  literature to our paper is~\cite{halman2014fully}. This papers also obtain a PTAS for some specific stochastic dynamic program similar to the Bayesian online allocation; however all these papers diverge from our treatment both in terms of techniques, results, and the category of the problems they can solve. 

Since an early conference version of our work\citep{anari2019nearly}, there has been a growing line of research on studying the optimum online benchmark in the Bayesian online allocation,  which is the same type of benchmark we consider in this paper.\cite{papadimitriou2021online} studies a slight variant of the matching prophet inequality problem, establishes PSPACE-hardness of computing optimum online, and obtains improved competitive ratios with respect to the optimum online benchmark.  There are also a limited number of other recent papers that consider competing with the optimum online benchmark in other combinatorial settings related to prophet inequalities~ \citep{ezra2022significance,braverman2022max}.

Finally, our work can also be considered as part of the rich literature on dynamic pricing in revenue management with inventory constraints. The classic work of \cite{gallego1994optimal} initiated the study of dynamic pricing with a given number of copies of the item to sell (limited supply), when the demand is stochastic (with known distribution) and price sensitive. Dynamic pricing in the i.i.d. stochastic setting when the demand distribution is unknown is also well studied~\citep[e.g., see][]{besbes2009dynamic,babaioff2012dynamic,babaioff2015dynamic}. See \cite{bitran2003overview} for different pricing models, and \cite{den2015dynamic} for a comprehensive survey on more recent results. Our work diverges from all above by considering the more complex combinatorial constraint of laminar matroid versus the limited supply.
Other indirectly related lines of work are bandits with knapsacks~\citep{badanidiyuru2013bandits,agrawal2016linear,immorlica2018adversarial},   online packing LP and convex optimization~\citep{devanur2011near,agrawal2014dynamic, agrawal2014fast,besbes2015non,li2022online,balseiro2023best}, the line of work on Bayesian prophet and low-regret framework for Bayesian online decision making~\citep{vera2018bayesian,vera2021online,banerjee2020uniform,kerimov2021dynamic}, and the growing literature on dynamic auctions and mechanism design~\citep{vulcano2002optimal,aviv2008optimal,gallien2006dynamic,balseiro2017dynamic}. Our paper diverges from all of these papers in terms of problem formulation and the underlying technical framework, and hence our results are not mathematically comparable to similar-in-spirit results in these papers.

\subsection{Organization}

The rest of the paper is organized as follows. In \Cref{sec:laminar} we introduce the laminar matroid Bayesian selection problem and provide a PTAS for the case where the depth of the laminar matroid is constant. In \cref{sec:production-simple}, we formalize the production constrained Bayesian selection problem (which is a special case of the setting in \Cref{sec:laminar}) and show how we can leverage the structure of this problem to go beyond constant depth. We further showcase the applications of our linear programming based techniques for the single-item prophet inequality problem in Appendix~\ref{sec:lp-prophet}. Finally, we have the concluding remarks and future directions, along with a list of open questions, in \Cref{sec:conclusion}.

\section{Laminar Matroid Bayesian Online Selection}
\newcommand{\lamfam}{\mathscr{F}}
\label{sec:laminar}
The goal of this section is to first introduce \emph{laminar matroid Bayesian selection} \citep{kleinberg2012matroid,feldman2016online}, and then propose a PTAS for the optimal online policy for maximizing social-welfare. On our way to achieve this goal, we will discuss an exponential-sized dynamic program and how it can be written as a linear program. We further relax this linear program to be able to solve it in polynomial time and then explore how it can be rounded to a feasible online policy without a considerable loss in expected social-welfare. The combination of these two ideas gives us our first polynomial time approximation scheme.

\subsection{Problem description}

The laminar matroid Bayesian selection is a special case of the well-known matroid Bayesian online selection problem studied in \cite{kleinberg2012matroid}. In this setting, we have a sequence of $n$ elements that arrive over time in an arbitrary but known order. \revcolor{Just before arrival, each element reveals its value.} We assume that the values are drawn independently from known heterogeneous distributions. More precisely, we assume the value of the element arriving at time $t$ is drawn from distribution $F_t$. Throughout this section, in order to have a succinct representation of the input for running time purposes, we focus on atomic distributions.

The goal is to design an online algorithm for picking a subset of these arriving elements that maximizes the expected social welfare, i.e. the expected sum of the values corresponding to the picked elements. Upon the arrival of each element, and after observing its value, the online algorithm needs to make an irrevocable decision about whether to pick or ignore the element. At the end, we want the set of picked elements to be feasible. The collection of feasible subsets are characterized by a given matroid $\mathcal{M}$.

In this paper, we consider the case where the feasible subsets are given by a special case of matroids called laminar matroids. More precisely, denote the set of all element by $E$. Consider a laminar family of subsets over these elements, i.e. a collection $\lamfam$ of subsets, termed as \emph{bins}, where for every $B,B'\in\lamfam$ either $B\subseteq B'$, $B'\subseteq B$ or $B\cap B'=\emptyset$. Each bin $B\in \lamfam$ has a capacity $k_B$ and we say a set $S\subseteq E$ is feasible if for each $B\in \lamfam$, we have $|S\cap B|\leq k_B$. It is often helpful to represent the laminar family as a rooted tree whose internal nodes are the bins and the leaves are the elements (without loss of generality, we assume the graph corresponding to our laminar family is a tree with a root corresponding to the largest set in $\lamfam$. Otherwise we can decompose the problem into smaller and independent subproblems each of which has this property). The depth of this tree represents the depth of our laminar matroid.

Throughout the paper, we will focus on characterizing the optimal \emph{online} policy and will evaluate our algorithms against that benchmark. In that sense, we deviate from the prophet inequality framework that compares various policies against the optimum offline. It is not hard to see that these two benchmarks could be off by a factor 2 of each other even for the special case of single item prophet inequality (see also \cite{kleinberg2012matroid}). Our main result in this section is a PTAS for the optimal online policy, when the depth of the family (or equivalently the height of the tree) is constant. We also show that our final algorithm has the form of an adaptive pricing with randomized tie-breaking. 

\subsection{Sketch of our approach}
\label{sec:approach}

One key idea in our approach is to mark each bin in our laminar family as either \emph{large} or \emph{small} and then treat each group differently in our analysis. 
We proceed with the following steps:
\begin{enumerate}
\item Finding a linear programming formulation (with exponential size) for characterizing the optimum online policy for this problem.
\item Developing a family of linear programming relaxations for the laminar matroid Bayesian selection problem. This family of relaxations is parametrized by how we \emph{mark} bins as large or small; we enforce the small bin capacities to be respected point-wise, while we allow the large bin capacities to hold in expectation. In this way, we essentially create a hierarchy of LP relaxations, where at the top of the hierarchy we have the \emph{expected relaxation}, a relaxation where all the capacities are allowed to be satisfied only in expectation, and at the bottom of the hierarchy we have an LP characterization of the optimum online policy. Importantly, all these linear programs can be solved up-front (i.e., offline); however they might not be solvable in polynomial time (see \cref{sec:lpformulation}).
\item Designing an adaptive pricing with randomized tie breaking policy to round the solution of any given such LP relaxation. We show the expected welfare of our policy is equal to the objective value of the particular LP relaxation it has started with, and so it is a lossless randomized rounding. Further, as expected, this solution respects all the small bin capacities of the LP relaxation point-wise and all the large bin capacities only in expectation. 
\item Presenting a particular \emph{marking algorithm} to select a polynomially solvable linear programming relaxation in the above mentioned hierarchy.
\item Using a concentration argument to show that the constraints corresponding to large bins are violated with only a small probability.
\end{enumerate}
We next elaborate more on each of the bullets above.

\subsection{Linear programming formulation of the optimum online policy}
\label{sec:lpformulation}
Our laminar matroid Bayesian selection problem can be solved exactly using a simple exponential-sized dynamic program. Let $\vec{s}\in \mathbb{Z}^{\lamfam}$ be the vector representing the number of picked elements in each bin of $\lamfam$. We say $\vec{s}$ is a \emph{feasible state} at time $t$ if it can be reached at time $t$ by a feasible online policy respecting the capacity constraints of all bins.

Define $\V_t(\vec{s})$ to be the maximum total expected welfare that an online policy can obtain from time $t$ to time $n$ given $\vec{s}$. Define $\V_{t}(\vec{s}) = -\infty$ when $s$ is not feasible at time $t$ and $\V_{n+1}(\vec{s}) = 0$ for all $\vec{s}$.  We can compute $\V_t(s)$ for the remaining values of $s$ and $t$ recursively as follows. At time $t$, the policy offers the buyer the price $\tau = \tau_t(\vec{s})$. Depending on whether or not the value of the customer is above $\tau$, the mechanism  obtains either  $v_t+\V_{t+1}(\vec{s}+\vec{d}_{t})$ or $\V_{t+1}(\vec{s})$, where $\vec{d}_t\in \{0,1\}^{\lamfam}$ is a binary vector denoting which bins in $\lamfam$ will be used if we pick the element arriving at time $t$\footnote{i.e. for every $B\in \lamfam$, $\vec{d}_t(B)=1$ if and only if $t\in B$.}. The probability of each event can be computed using the distribution of the value of element $t$.  Therefore, the dynamic programming table can be computed using the following rule also known as the \emph{Bellman equation}:
\begin{equation}
\label{eq:bellman}
\V_{t}(\vec{s})= \max_\tau \parens*{\ExX{v_t\sim F_t}{\parens*{v_t+\V_{t+1}(\vec{s}+\vec{d}_{t})}\cdot \1{v_t\geq \tau}}+\ExX{v_t\sim F_t}{\V_{t+1}(\vec{s})\cdot \1{v_t<\tau}}}.
\end{equation}

Note that the price $\tau_t(\vec{s})= \V_{t+1}(\vec{s})-\V_{t+1}(\vec{s}+\vec{d}_{t})$ maximizes the above equation, and so the final prices of an optimal online policy can be computed easily given the table values.  

The above dynamic program has an exponentially large table.  In the rest of this section, we describe a linear programming formulation equivalent to the above dynamic program, a natural relaxation for the LP, and a randomized rounding of the relaxation that yields a PTAS when the depth of the laminar family is constant.

\label{sec:lp-optimal}

An online policy can be fully described by \emph{allocation variables} $\X_{t}(\vec{s},v)$, where for every time $t$ and state $\vec{s}$, $\X_{t}(\vec{s},v)$ represents the probability of the event that the element arriving at time $t$ is picked and the state upon its arrival is $\vec{s}$, conditioned on $v_t=v$. We further use  \emph{state variables} $\Y_{t}(\vec{s})$ to represent the probability of the event that an online policy reaches the state $\vec{s}$ upon the arrival of the element at time $t$, and auxiliary variables $\X_{t}(v)$ for the marginal probability of picking the element arriving at time $t$ conditioned on $v_t=v$. Using these variables we can write the optimum online policy as a linear program. A similar formulation has been used in \citet{niazadeh2018prophet} to characterize the optimum online policy.

Having this description, the LP formulation of the above dynamic program is a combination of two new ideas. \revcolor{The first idea is to ensure the feasibility of the policy by adding the constraint that $\Exv{v_{t-1}}{\X_{t-1}(\vec{s}, v_{t-1})}=0$ for any feasible state $\vec{s}$ at any time $t-1$ in which by an allocation at time $t-1$ we lead to an infeasible state at time $t$. This, along with starting from a feasible state, will automatically ensure $\Y_{t}(\vec{s})=0$ for any infeasible state $\vec{s}$ at any time $t$.} The second idea is to add another constraint describing how the probability $\Y_{t}(\vec{s})$ updates from time $t$ to $t+1$ as the result of the probabilistic decision made by the policy at time $t$. As will be elaborated more later, this constraint is the necessary and sufficient condition for any policy to be implementable in an online fashion.

Let the set $\mathcal{S}\subset \mathbb{Z}^{\lamfam}$ be a finite set containing all possible feasible states at any time $t$.\footnote{For the ease of exposition, we do not consider time-specific state spaces. In particular, let $\mathcal{S}$ to be the set of all possible states that can happen by picking a subset of elements of size at most $K$. This set contains $O(n^K)$ states, where at any time $t$ only a subset of them are actually reachable.} Consider the following exponential-sized (both in the number of variables and constraints) linear program: 

 



 
\begin{equation*}
\label{eq:lp-optimal}\tag{LP$_1$}
\begin{array}{ll@{}ll}
\text{maximize}  & \displaystyle \sum_{t=1}^{n}\Exv{v_t}{v_t\cdot\X_t(v_t)} &\\
\text{subject to}& \{\X_t(\vec{s},v),\X_t(v),\Y_t(\vec{s})\}\in \Po^{\textrm{opt}}~,
\end{array}
\end{equation*}
where $\Po^{\textrm{opt}}$ is the polytope of \emph{point-wise feasible online policies}, defined by these linear constraints:
\begin{equation*}
\begin{array}{ll@{}ll}
& \X_t(v)=\displaystyle\sum_{\vec{s}\in \mathcal{S}}\X_{t}(\vec{s},v) &~~\forall v,~ t=1,2, \ldots, n,\\ \\
 &0 \leq \X_{t}(\vec{s},v)\leq \Y_{t}(\vec{s})&~~\forall v,~\vec{s}\in\mathcal{S}, t=1,2, \ldots, n, \\
&\Y_{t+1}(\vec{s})=\Y_{t}({\vec{s}})-\Exv{v_t}{\X_{t}(\vec{s}, v_t)}+\Exv{v_t}{\X_{t}(\vec{s}-\vec{d}_{t}, v_t)}&~~\forall \vec{s}\in \mathcal{S},~ t=1, \ldots, n&~({\emph{\textrm{state update}}})  \\ \\
 & \revcolor{\Y_1(\vec{0})\leq 1~,~\Y_1(\vec{s})=0} &~~\revcolor{\forall \vec{s}\in\mathcal{S}\setminus\{\vec{0}\}} \\
&\revcolor{{\X_{t}(\vec{s}, v})=0}&\revcolor{~~\forall v, t=1,\ldots,n, ~\vec{s}\in \mathcal{S}: \vec{s}+\vec{d}_t\in \partial \mathcal{S}}&\revcolor{~({\emph{\textrm{feasibility check}}})}
\end{array}
\end{equation*}
where, as a reminder, $\vec{d}_t\in \{0,1\}^{\lamfam}$ is a binary vector denoting which bins in $\lamfam$ will be used if we pick the element arriving at time $t$. We also use $\partial \mathcal{S}$ to denote the set of all \emph{forbidden neighboring states} at time $t$. 
$$\partial S\triangleq\left \{\vec{s}\in \mathbb{Z}^{\lamfam}: \left[\textrm{$\vec{s}$ is not a feasible state}\right]~\&~\left[\exists t ~\textrm{s.t.}~\vec{s}-\vec{d}_{t}\in \mathcal{S}\right]\right\}$$

It is also not hard to see that any feasible online policy induces a feasible assignment for the linear program (\ref{eq:lp-optimal}). The only tricky constraint to check is the constraint corresponding to the ``state update''. To do so, note that the online policy will reach the state $\vec{s}$ at time $t+1$  if and only if either the state at time $t$ is $\vec{s}$ and the element arriving at time $t$ is not picked, or the state at time $t$ is $\vec{s}-\vec{d}_{t}$ and the element arriving at time $t$ gets picked, evolving the state from $\vec{s}-\vec{d}_{t}$ to $\vec{s}-\vec{d}_{t}+\vec{d}_{t}=\vec{s}$. 

More importantly, we show the converse holds by proposing an exact rounding algorithm in the form of an \emph{adaptive pricing with randomized tie-breaking} policy; such a policy sets a price $\T_t(\vec{s})$ for the element arriving at time $t$ if the current state is $\vec{s}$. In case of a tie ($v_t=\T_t(\vec{s})$), the pricing policy breaks the tie independently with probability $p_t(\vec{s})$, in favor of selling the item.


\begin{proposition}
\label{prop:adaptive-pricing}
There exists an adaptive pricing policy with randomized tie breaking, whose expected social-welfare is equal to the optimal solution of the linear program~(\ref{eq:lp-optimal}) and is a feasible online policy for the laminar matroid Bayesian selection problem.
\end{proposition}

We postpone the formal proof and a discussion on how to compute prices and tie-breaking probabilities (given the optimal solution to LP) to \cref{sec:proof-of-prop1} and just sketch the main ideas here. 

\proof{Proof sketch.} Let $\{\X^*_t(\vec{s},v)\}$ and $\{\Y^*_t(\vec{s})\}$ be the optimal solutions of \ref{eq:lp-optimal}. Consider the following simple online randomized rounding scheme: start from the all-zero assignment at time $t<1$. Now, suppose at time $t\geq 1$, the current state, i.e., number of sold products of different types, is $\vec{s}$ and the realized value of the arriving element is $v_t=v$. The rounding algorithm first checks whether $\Y^*_t(\vec{s})$ is zero. If yes, it skips the element. Otherwise, it picks the element with probability $ \tfrac{\X^*_t(\vec{s},v)}{\Y^*_t(\vec{s})}$. 

It is not hard to show this simple scheme will have allocation and state probabilities matching the LP optimal assignment, i.e. $\{\X^*_t(\vec{s},v)\}$ and $\{\Y^*_t(\vec{s})\}$. Moreover, $\Y^*_t(\vec{s})=0$ for all forbidden neighboring states $\vec{s}$, i.e. infeasible states that can only be reached from a feasible state at time $t$ by accepting an extra request. Hence an inductive argument shows that the resulting online policy is always feasible. There is also a simple coupling argument, with shifting the probability masses to higher values, showing that the above algorithm can be implemented using an adaptive pricing policy with the randomized tie breaking. Prices and probabilities can then be computed by straightforward calculations. \qed
\endproof

\subsection{A hierarchy of linear programming relaxations for general laminar matroids}
\label{sec:lp-hierarchy-laminar}
We define a family of linear programming relaxations, parametrized by different markings of bins in small and large. \revcolor{It is important to note that our marking is hereditary, meaning that we mark bins in a way that the child of a small bin is always small and the parent of a large bin is always large. Given a particular feasible marking as described (see also \cref{sec:approach})}, let $\La$ be the set of large bins and $\Sm$ be the set of maximal small bins.


We can compute the optimum online policy using the (exponential time) linear program~(\ref{eq:lp-optimal}). To avoid exponentially many states in our hierarchy of LP relaxations, we use the same state-space structure, but we only track the \emph{local state} of maximal small bins in $\Sm$ separately. In other words, we can think of each maximal small bin $B$ as a separate laminar matroid Bayesian selection sub-problem with laminar family $\lamfam^B\triangleq \{B'\in\lamfam:B'\subseteq B\}$, where each arriving element is only in one of the sub-problems (because subsets in $\Sm$ form a partition of the set of all elements). Now, if the arriving element at time $t$ belongs to $B\in\Sm$, the linear program only needs to keep track of the change in the local state $\vec{s}\in\mathbb{Z}^{\lamfam^B}$ of the sub-problem $B$, i.e. the vector representing the number of picked elements of each bin in $\lamfam^B$.

For every small bin $B$, define $\mathcal{S}^B$ to be the set of all \emph{feasible local states} of the sub-problem $B$, i.e. the set of all possible states that can be reached by an online policy for this sub-problem that respects all the capacities in $\lamfam^B$. Note that $\lvert\mathcal{S}^B\rvert \leq n^{k_B}$, because no feasible online policy for the sub-problem $B$ can pick more than $k_B$ elements. We now can write a linear program with the following variables and constraints:

\paragraph{\textbf{Variables.}}We add \emph{allocation variables} $\X_t(\vec{s},v)$, \emph{marginal allocation variables} $\X_t(v)$ and \emph{state variables} $\Y_t(\vec{s})$ as before. For the variables $\X_t(\vec{s},v)$ and $\Y_t(\vec{s})$, assuming the element arriving at time $t$ belongs to the maximal small bin $B$, the vector $\vec{s}$  represents the local state of $B$ right before arrival of this element.
\paragraph{\textbf{Constraints.}} We add two categories of linear constraints to our LP relaxations:
\begin{itemize}
\item \emph{Global expected constraints}: these constraints ensure that the capacity of all large bins are respected in expectation, i.e.
\begin{equation*}
\forall B\in \La: \displaystyle \sum_{t\in B}\Exv{v_t}{\X_t(v_t)}\leq k_B
\end{equation*} 
\item \emph{Local online feasibility constraints:} for every bin $B\in \Sm$, similar to \ref{eq:lp-exante}, we can define a polytope $\Po^B$ of feasible online policies that ensures a feasible assignment of the linear program is online implementable by a feasible policy. So, these constraints will be:
\begin{equation*}
\forall B\in\Sm: \{\X_t(\vec{s},v), \X_t(v), \Y_t(\vec{s})\}\in \Po^B
\end{equation*}
\end{itemize}
\paragraph{Polytope of feasible online policies.}
The polytope $\Po^B$ is defined using exactly the same style of linear constraints as in \Cref{sec:lpformulation} :
\begin{equation*}
\begin{array}{ll@{}ll}
& \X_t(v)=\displaystyle\sum_{\vec{s}\in \mathcal{S}^B}\X_{t}(\vec{s},v) &~~~~\forall v,~ t\in B,\\ \\
                                 
  &0 \leq \X_{t}(\vec{s},v)\leq \Y_{t}(\vec{s})&~~~~\forall v,~\vec{s}\in\mathcal{S}^B, t\in B, \\
&\Y_{t}(\vec{s})=\Y_{t'}({\vec{s}})-\Exv{v_{t'}}{\X_{t'}(\vec{s}, v_{t'})}+\Exv{v_{t'}}{\X_{t'}(\vec{s}-\vec{d}_{t'}, v_{t'})}&~~~~\forall \vec{s}\in \mathcal{S}^B,~ t,t'\in B,~~~~~~~~~~~~~~~~~~({\emph{\textrm{state update}}})\\ 
&&~~~~[t'+1:t-1]\cap B=\emptyset \\ \\
& \revcolor{\Y_{t_0}(\vec{0})\leq 1~,~\Y_{t_0}(\vec{s})=0}&\revcolor{~~~~\forall \vec{s}\in\mathcal{S}^B\setminus\{\vec{0}\}~,~ t_0=\min\{t\in B\}}\\
&\revcolor{{\X_{t}(\vec{s}, v)}=0}&\revcolor{~~~~\forall v,~\vec{s}\in \mathcal{S}^B: \vec{s}+\vec{d}_t\in \partial \mathcal{S}^B, t\in B~({\emph{\textrm{feasibility check}}})}
\end{array}
\end{equation*}
where $\vec{d}_t\in \{0,1\}^{\lamfam^B}$ is a binary vector denoting which bins in $\lamfam^B$ will be used if we pick the element arriving at time $t\in B$, and $\partial \mathcal{S}^B$ is the set of all \emph{forbidden neighboring states} of sub-problem $B$, i.e. 
$$\partial S^B\triangleq \{\vec{s}\in \mathbb{Z}^{\lamfam^B}: \left[\textrm{$\vec{s}$ is an infeasible local state}\right]~~\&~~\left[\exists t\in B,~\textrm{s.t.}~ \vec{s}-d_t\in \mathcal{S}^B\right] \}.$$

It is easy to see that the set $\partial \mathcal{S}^B$ has at most $O(n^{k_B+1})$ states. Given these variables and constraints, the LP relaxation corresponding to the marking $(\Sm,\La)$ (which we show in \Cref{prop:lp-hierarchy-relaxation-laminar} why is actually a relaxation) can be written down as follows.

\begin{equation*}
\label{eq:lp-hierarchy-laminar}
\begin{array}{ll@{}ll}
\text{maximize}  & \displaystyle \sum_{t=1}^{n}\Exv{v_t}{v_t\cdot\X_t(v_t)}&\tag{LP$_2$}&\\
\text{subject to}&  \displaystyle \sum_{t\in B}\Exv{v_t}{\X_t(v_t)}\leq k_B&~~~~\forall B\in \La&\textit{(Global expected constraints)}, \\
& \\
														& \{\X_t(\vec{s},t),\Y_t(\vec{s}),\X_t(v)\}\in \Po^B &~~~~\forall B\in\Sm&\textit{(Local online feasibility constraints).} 
\end{array}
\end{equation*}

Again, it is easy to see that any feasible online policy for the sub-problem $B\in \Sm$ is represented by a feasible point inside the polytope $\Po^B$. As every online policy for the laminar matroid Bayesian selection problem induces a feasible online policy for each sub-problem $B\in\Sm$ (by simulating the randomness of the policy and values outside of $B$), and because it respects all the large bin capacity constraints point-wise, we have the following proposition.


\begin{proposition}
\label{prop:lp-hierarchy-relaxation-laminar}
For any marking $(\Sm,\La)$ of the laminar tree, \ref{eq:lp-hierarchy-laminar} is a relaxation of the optimal online policy for maximizing expected social-welfare in the laminar matroid Bayesian selection problem.
\end{proposition}
\begin{proof}{Proof.}
Consider the optimal online policy and its induced feasible online policy for the sub-problem $B$. Let $\{\X_t(\vec{s},v_t)\}_{t\in B}$ be the allocation probabilities and $\{\Y_t(\vec{s})\}_{t\in B}$ be the state evolution probabilities of this policy. First of all, clearly the objective function of the LP is equal to the expected social welfare of the online policy. Second, $\Y_{t_0}([k_{B'}]_{B'\in\lamfam^B})=1$, as the policy has not yet picked any elements in $B$ when the first element in $B$ arrives. Moreover, the policy respects all the capacity constraints point-wise. Hence, in the resulting assignment $\Y_t(\vec{s})=0$ for $\vec{s}\in\partial\mathcal{S}^B$ and the global expected constraints are satisfied.

The only remaining constraint to check is the Bellman update constraint of $\Po^B$. In order to see the satisfaction of the constraint, note that the policy will reach state $\vec{s}$ at time $t$ if and only if either the state at time $t'$ (i.e. the last time an element arrived in $B$) is $\vec{s}$ and the element at time $t'$ is not selected, or the state at time $t$ is $\vec{s}+\vec{d}_{t'}$ and the element at time $t'$ is selected, evolving the state from $\vec{s}+\vec{d}_{t'}$ to $\vec{s}+\vec{d}_{t'}-\vec{d}_{t'}=\vec{s}$. Therefore, the state evolution probabilities satisfy the Bellman update constraint.
\qed
\end{proof}

\subsection{Exact rounding through adaptive pricing with randomized tie-breaking}
Given a particular marking $(\Sm,\La)$, we show there exists a family of adaptive pricing with randomized tie-breaking policies where each of these pricing policies exactly rounds the solution induced by the optimal solution of (\ref{eq:lp-hierarchy-laminar}) in each small bin $B\in \Sm$.

The above rounding schemes can then be combined with each other, resulting in an online policy that is point-wise feasible inside each small bin and only feasible in expectation inside each large bin, i.e. it only respects the large bin capacity constraints in expectation. Our randomized tie-breaking policies are characterized by $\left\{\tau_t^B(\vec{s}),p^B_t(\vec{s})\right\}_{B\in\Sm}$. The final procedure is simple: once an element arrives at time $t$ that belongs to $B\in \Sm$, the algorithm looks at the state of the bin $B$ (suppose it is $\vec{s}$), and posts the price $\tau^B(\vec{s})$ with tie-breaking probability $p_t(\vec{s})$. The element is then accepted w.p. 1 if $v_t>\tau_t^B(\vec{s})$, w.p. 0 if $v_t<\tau_t^B(\vec{s})$, and w.p. $p_t(\vec{s})$ if $v_t=\tau_t^B(\vec{s})$. We formalize this discussion in the following proposition. 


\begin{proposition}
\label{prop:adaptive-pricing-laminar}
For the laminar Bayesian online selection problem, given any marking $(\Sm,\La)$ of the laminar tree, there exists an adaptive pricing policy with randomized tie breaking whose expected  welfare is equal to the optimal solution of the linear program~(\ref{eq:lp-hierarchy-laminar}). Moreover, the resulting policy is feasible inside each small bin and feasible in expectation inside each large bin.
\end{proposition}

By putting all the pieces together, we run the following algorithm given a particular marking.
\begin{algorithm}[ht]
\small
\algblock[Name]{Start}{End}
\algblockdefx[NAME]{START}{END}%
[2][Unknown]{Start #1(#2)}%
{Ending}
\algblockdefx[NAME]{}{OTHEREND}%
[1]{Until (#1)}
 \caption{ PTAS-Laminar~($\Sm,\La,\eps$)
     \label{alg:ptas-laminar}}
      \begin{algorithmic}[1]
        \State{\textbf{Input}} parameter $\eps>0$.
        \State Multiply the capacities of all the large bins $B\in\La$ by $(1-\eps)$. 
        \item Solve the LP relaxation (\ref{eq:lp-hierarchy-laminar}) for the given marking $(\Sm,\La)$.
        \State Extract adaptive prices $\{\tau^B_t(\vec{s})\}$ and adaptive tie-breaking probabilities $\{p^B_t(\vec{s})\}$ for every maximal small bin $B\in \Sm$ and $t\in B$.
        \State Run the adaptive pricing with randomized tie breaking for each small bin $B\in \Sm$ separately, using the computed prices and probabilities in step (4). In the exceptional cases when there is no remaining capacity when a customer arrives, we offer her the price of infinity. 
      \end{algorithmic}
\end{algorithm}
\begin{remark}
Once an element $t$ arrives, the algorithm identifies the maximal small bin $B\in\Sm$ that contains $t$, and finds the current state $\vec{s}$ in this bin. It then posts the price $\tau_t^B(\vec{s})$ with randomized tie-breaking probability $p^B_t(\vec{s})$.
\end{remark}
\begin{proof}{{Proof of \Cref{prop:adaptive-pricing-laminar}}.}
 Similar to \Cref{prop:adaptive-pricing}, Let $\{\X^*_t(\vec{s},v)\}$ and $\{\Y^*_t(\vec{s})\}$ be the optimal solutions of \ref{eq:lp-hierarchy-laminar}. Consider the following simple online randomized rounding scheme: start at time $0$ where no elements are picked. Now, suppose at time $t\geq 1$, the current state, i.e., number of picked elements in each bin is represented by $\vec{s}$ and the realized value of the arriving element is $v_t=v$. The rounding algorithm first checks whether the arriving element belongs to a small bin. If it does not, then it picks the element with probability $\X_t(v_t)$. On the other hand, if the arriving element belongs to some small bin, the rounding algorithm first checks whether $\Y^*_t(\vec{s})$ is zero. If yes, it skips the element. Otherwise, it picks the element with probability $ \tfrac{\X^*_t(\vec{s},v)}{\Y^*_t(\vec{s})}$. Details of the proof are similar to that of \Cref{prop:adaptive-pricing} and hence are omitted for brevity. \qed
\end{proof}
\subsection{Marking and concentration for constant-depth laminar}
\label{sec:laminar-mark}
In this section, we want to show that our rounding algorithm (\cref{alg:ptas-laminar}) achieves a $(1-O(\eps))$ fraction of the expected social-welfare obtained by the optimal online policy. Note that (\ref{eq:lp-hierarchy-laminar}) is a relaxation, and scaling down the large capacities by a factor of $1-\epsilon$ reduces the benchmark by at most a factor $(1-\eps)$ \revcolor{(simply because if we pick an optimal solution $\{\X^*_t(\vec{s},v)\}$ and $\{\Y^*_t(\vec{s})\}$  of \ref{eq:lp-hierarchy-laminar} before reducing the capacities, and then multiply this solution by $(1-\eps)$, the objective value is multiplied by $1-\epsilon$, while this new solution will become feasible in the modified version of the LP with capacities reduced by a factor of $1-\epsilon$).}

Once an element arrives at time $t$, consider all large bins $B\in \La$ that are along a path from this element to the root of the laminar tree. By construction, the expected value extracted from this element by the pricing policy would be exactly equal to the contribution of this element to the objective value of (\ref{eq:lp-hierarchy-laminar}) (after scaling down the capacities), but only if the element is not ignored; an element will be ignored, i.e., offered a price of infinity, if one of the mentioned large capacities is exceeded. Therefore, to show that the loss is bounded by $O(\eps)$ fraction of total, we only need to show that the bad event of an element being ignored happens with a probability that is bounded by $O(\eps)$.

To bound the above probability, we need a concentration bound for the random variable corresponding to the total number of elements picked in each large bin. If we could show negative dependence among selection indicators of the optimal online policy (with a particular ordering of the elements), or if we show negative dependency between the indicator random variable of selecting an element and the number of selected elements so far, we could get a concentration using the Chernoff bound or Azuma inequality for super-martingales. In fact, this is something we will exploit in the next section for a subclass of laminar matroids. 

Nevertheless, the particular forms of negative dependence above do not hold for general laminar matroids with arbitrary arrival order of elements. We show this fact in the following example.
\begin{figure}[H]
  \centering
  \includegraphics[width=0.35\columnwidth]{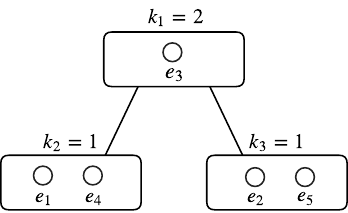}
  \vspace{2mm}
     \caption{Bad example showing lack of negative dependency for optimal online policy of general laminar matroid. }
     \label{fig:example1}
\end{figure}
\begin{example}
	Consider the laminar matroid depicted in \Cref{fig:example1} with elements arriving one by one from $e_1$ to $e_5$. 
Let $e_4$ and $e_5$ be uniformly distributed on $\set{0, 2}$ and $e_3$ be uniformly distributed on $\set{0, 1}$. If $e_1$ is picked, then only one of $e_2, e_3, e_5$ can be picked. One can see that in this case, the price offered to $e_2$ would be 1 since $\mathbb{E}[\max(e_3, e_5)]=1$. On the other hand, if $e_1$ is discarded, there are two cases. (i) if $e_2$ is discarded, then the optimum online policy would have to pick two of $e_3, e_4, e_5$. One can see that the expected value obtained by the optimum online policy is 2.25. (ii) If $e_2$ is selected, then the optimum online policy would have to select one item from $e_3, e_4$, in which case the expected obtained value would be $1$. Therefore, in this case, the price offered to $e_2$ would have to be 1.25.
		
	Note that this example shows that by not selecting $e_1$, there is a higher price offered to $e_2$ which means, by definition, that the negative dependency does not hold (neither between indicator random variables corresponding to selection of different elements nor between the indicator random variable of selecting an element and the number of selected elements in the past).

\end{example}


We now propose a marking algorithm, parametrized by $\delta>0$, such that it guarantees the required concentration. Without loss of generality, assume $k_B\leq k_{B'}$ for any two bins $B$ and $B'$, whenever $B$ is a child of $B'$ in the laminar tree.\footnote{Otherwise, just drop the constraint on the child.} Let $L$ be the depth of the given instance (which we assume is constant in this section). Now, for every bin $B$ at depth $d$ of the laminar tree, i.e. when it has distance $d$ from the root,  mark it as small if and only if $k_B\leq \frac{1}{\delta^{L-d}}$, and large otherwise. If a node is marked as small, then we mark all of its children as small too (\Cref{fig:marking}).

\begin{figure}[H]
  \centering
  \includegraphics[width=0.9\columnwidth]{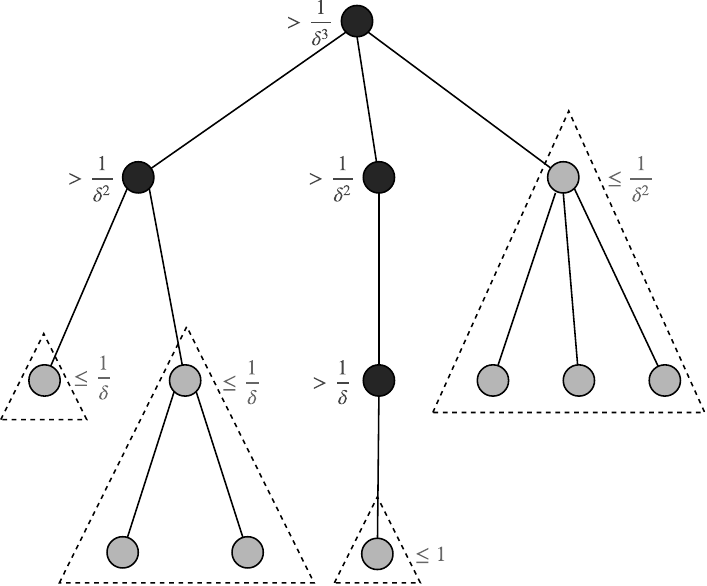}
     \caption{Depth-based marking for concentration. }
     \label{fig:marking}
\end{figure}

The key idea here is that our proposed marking algorithm provides \emph{enough separation} between a large bin and its small descendants. In fact, we partition the bins so that the capacity of every large bin is at least $\frac{1}{\delta}$ times the capacity of any of its small immediate descendant bins. This separation provides us with the required concentration bound. 
\begin{theorem}
\revcolor{Using the proposed marking algorithm and by setting $\delta=\frac{\eps^2}{3\log(L/\eps)}$, \cref{alg:ptas-laminar}}  is a $(1-O(\eps))$-approximation for the expected welfare of the optimal online policy in the laminar matroid Bayesian selection problem with depth $L$, and runs in time $\texttt{poly}(n)$ assuming $L$ and $\eps$ to be constant.
\end{theorem}
\begin{proof}{Proof.}
\revcolor{Consider a hypothetical run of \Cref{alg:ptas-laminar} on maximal small bins in $\Sm$ ignoring the capacity constraints corresponding to large bins in $\La$.} For every $B'\in\Sm$, let $C_{B'}$ denote the total number of elements picked from this bin. As maximal small bins $\Sm$ induce partitioning over the set of all elements and since \cref{alg:ptas-laminar} runs an independent online policy in each bin $B'\in \Sm$, random variables $\{C_{B'}\}_{B'\in\Sm}$ are mutually independent. 

Now consider a large bin $B\in\La$ and let $B'_1,...,B'_m$ be the maximal small descendant bins that partition $B$. Moreover, assume that bin $B$ is at depth $d$ of the laminar tree. Hence $k_B>\frac{1}{\delta^{L-d}}$ and $k_{B'_{j}}\leq \frac{1}{\delta^{L-d-1}}$ for $j=1,...,m$. As the policy inside each small bin $B'_j$ is a point-wise feasible policy, $C_{B'_j}\leq k_{B'_{j}}\leq \frac{1}{\delta^{L-d-1}}$. Moreover, linear program~(\ref{eq:lp-hierarchy-laminar}) imposes a soft-constraint equal to $\sum_{t\in B'_j}\Exv{v_t}{\X^*_t(v_t)}$ on each small bin. This soft-constraint should be respected in expectation by the final pricing policy (due to the construction of our randomized rounding scheme), so $\Ex{C_{B'_j}}\leq \sum_{t\in B'_j}\Exv{v_t}{\X^*_t(v_t)}$ for all $B'_j$. Because of the global expected constraints, these soft constraints should respect the large bin capacity constraints of the laminar matroid when large capacities are scaled down by a factor $(1-\eps)$. Therefore
\[
\displaystyle \sum_{j=1}^m\Ex{C_{B'_j}}\leq \sum_{j=1}^{m}\sum_{t\in B'_j}\Exv{v_t}{\X^*_t(v_t)}= \sum_{t\in B}\Exv{v_t}{\X^*_t(v_t)}\leq k_B(1-\eps)
\]
Now, by applying simple Chernoff bound for independent random variables $\{C_{B'_j}\}_{j=1}^m$, we have (define the notation $\tilde{C}_{B'_j}\triangleq C_{B'_j}\cdot\delta^{L-d-1}$, so that it normalizes the total count to $[0,1]$):
\begin{align}
\Pr{\displaystyle\sum_{j=1}^m C_{B'_j}>k_B}&=\Pr{\displaystyle\sum_{j=1}^m (C_{B'_j}-\Ex{C_{B'_j}})>k_B-\displaystyle\sum_{j=1}^m \Ex{C_{B'_j}}}\leq \Pr{\displaystyle\sum_{j=1}^m (C_{B'_j}-\Ex{C_{B'_j}})>\eps k_B}\nonumber\\
&= \Pr{\displaystyle\sum_{j=1}^m (\tilde{C}_{B'_j}-\Ex{\tilde{C}_{B'_j}})>\eps\cdot\delta^{L-d-1}k_B}      \leq \exp\left(  -\frac{\eps^2\cdot\delta^{2(L-d-1)}\cdot k_B^2}{3\sum_{j=1}^{m}\Ex{\tilde{C}_{B'_j}}}\right)\nonumber\\
&\label{eq:bin-bound}\leq\exp\left(  -\frac{\eps^2\cdot\delta^{(L-d-1)}\cdot k_B}{3}\right)\revcolor{\leq \exp\left(  -\frac{\eps^2}{3\delta}\right)~,}
\end{align}
\revcolor{where in the last inequality we use the fact that $K_B>\frac{1}{\delta^{L-d}}$ for a bin $B$ at depth $d$ in our marking algorithm. Now, for a particular element $t$, consider a path from the maximal small bin containing this element to the root of the laminar tree. This path may contain several large bins and we need to check if their capacities are exceeded at the time of arrival of the element $t$ in the hypothetical run of \Cref{alg:ptas-laminar}(when we ignore large-bin capacities). We take a union bound over all such bad events, noting that there are at most $L$ such bad events, simply because the length of the path connecting the maximal small bin to the root is at most the depth of the laminar matroid.
 \revcolor{For each element $t'$, let ${Z}_{t'}$  be the allocation binary variable of element $t'$ in this hypothetical run.}
 By applying union bound we have:
 \begin{equation*}
 \Pr{\exists B\in\La:t\in B,\displaystyle\sum_{t'\in B}Z_{t'}>k_{B}}\leq Le^{\left(  -\frac{\eps^2}{3\delta}\right)}=\eps
 \end{equation*}
where the first inequality is due to union bound (over at most $L$ bad events, each corresponding to one of the large bins on the path from the small bin to the root) and the upper bound established in \eqref{eq:bin-bound} for a large bin at any depth $d$, and the last equality holds as $\delta=\frac{\eps^2}{3\log(L/\eps)}$.


Consequently, with probability at least $1-\epsilon$ none of these capacities are exceeded at the time the algorithm processes element $t$ in this hypothetical run. Therefore, if for each element $t$ we compare the actual run of \Cref{alg:ptas-laminar} (when large bin capacities are enforced; see the description of the algorithm) with the hypothetical run (when large bin capacities are ignored),  in  $1-\epsilon$ fraction of sample paths we see that the two algorithms obtain exactly the same value from request $t$. Therefore, due to the linearity of expectations, \cref{alg:ptas-laminar} achieves $(1-\eps)$ fraction of the expected social-welfare of \Cref{alg:ptas-laminar} in this hypothetical run when all the large bin capacities are  ignored, which is exactly equal to the objective value of $\eqref{eq:lp-hierarchy-laminar}$ when large capacities are multiplied by $1-\eps$. As mentioned earlier, multiplying capacities by $1-\eps$ only reduces the objective value by a multiplicative factor of $1-\eps$. Putting all pieces together, and due to the fact that $\eqref{eq:lp-hierarchy-laminar}$ is a relaxation for the optimal online policy, \Cref{alg:ptas-laminar} obtains $(1-\eps)^2=1-O(\eps)$ fraction of the expected social-welfare of the optimal online policy. Moreover, the linear program~(\ref{eq:lp-hierarchy-laminar}) has size at most $O(n^{\frac{1}{\delta^L}})$, and hence the running time is $\texttt{poly}\left(n^{\left(\frac{3\log(L/\eps)}{\eps^{2}}\right)^L}\right)$ by setting $\delta=\frac{\eps^2}{3\log(L/\eps)}$. This running time is $\texttt{poly}(n)$ assuming $L$ and $\eps$ are constant. \qed}
\end{proof}

\subsection{Formal proof of \texorpdfstring{\Cref{prop:adaptive-pricing}}{TEXT} and adaptive prices/tie-breaking probabilities}
\label{sec:proof-of-prop1}
To end this section, we provide the details of how to round the linear programming solution of  (\ref{eq:lp-optimal}). Let $\{\X^*_t(\vec{s},v),\Y^*_t(\vec{s})\}$ be the optimal solution of LP. First consider the following simple online randomized rounding scheme. It starts at time $0$ with no elements picked. Now, suppose at time $t\geq 1$, the current state, i.e. vector representing the number of picked element in each bin, is $\vec{s}$ and the realized value of the arriving element is $v_t=v$. The rounding algorithm first checks whether $\Y^*_t(\vec{s})$ is zero. If yes, it skips the element. Otherwise, it flips a coin and with probability $\tfrac{\X^*_t(\vec{s},v)}{\Y^*_t(\vec{s})}$ picks the element. Note that if we assume for all possible states $\vec{s'}$, and all possible value $v'$ the following holds so far:
\begin{align*}
&\forall t'<t: \Pr{\textrm{policy picks the element arriving at time $t'$ and the state at $t'$ is $\vec{s'}$} \given v_{t'}=v'}=\X^*_{t'}(\vec{s'},v'),\\
&\forall t\leq t: \Pr{\textrm{policy reaches the state $\vec{s'}$ at $t'$}}=\Y^*_{t'}(\vec{s'}),
\end{align*}
then by progressing from $t$ to $t+1$, the same invariant holds because:
\begin{align*}
&\Pr{\textrm{policy picks the element arriving at time $t$ and the state at $t$ is $\vec{s}$} \given v_{t}=v}\\
=&\Pr{\textrm{policy reaches the state $\vec{s'}$ at $t'$}}\cdot \tfrac{\X^*_t(\vec{s},v)}{\Y^*_t(\vec{s})}=\Y^*_t(\vec{s})\cdot\tfrac{\X^*_t(\vec{s},v)}{\Y^*_t(\vec{s})}=\X^*_t(\vec{s},v)
\end{align*}
where in the first line we use that $v_t\sim F_t$ is independently drawn from the past. Moreover, 
\begin{align*}
&\Pr{\textrm{policy reaches the state $\vec{s}$ at $t+1$}}\\
&=\Pr{\textrm{policy reaches the state $\vec{s}$ at $t$}}-\Pr{\textrm{policy the element arriving at time $t$ and the state at $t$ is $\vec{s}$ }}+\\
&~~~~~\Pr{\textrm{policy picks the element arriving at time $t$ and the state at $t$ is $\vec{s}+\vec{d}_{t}$}}\\
&=\Y^*_{t}(\vec{s})-\Exv{v_t}{\X^*_t(\vec{s},v_t)}+\Exv{v_t}{\X^*_t(\vec{s}+\vec{u}_{l_t},v_t)}=\Y^*_{t+1}(\vec{s})
\end{align*}
where in the last line we used the fact that the optimal solution of the LP satisfies the state evolution update rule (i.e. the LP constraint). 

Putting all the pieces together, we conclude that the mentioned simple randomized rounding exactly \emph{simulates} the probabilities predicted by the optimal LP solution, and hence is a point-wise feasible online policy with the same expected social welfare as the optimal value of the LP. 

It only remains to show how an adaptive pricing mechanism with randomized tie breaking can also round the LP exactly. Let $\Z^*_t(\vec{s},v)\triangleq\tfrac{\X^*_t(\vec{s},v)}{\Y^*_t(\vec{s})}$ for every $v$ and $\vec{s}$ where $\Y^*_t(\vec{s})\neq 0$. By applying a simple coupling argument, we claim there exists a threshold $\tau$ such that $\Z^*_t(\vec{s},v)=1$ for $v>\tau$, and $\Z^*_t(\vec{s},v)=0$ for $v<\tau$. If not, one can slightly move the allocation probability mass of the randomized rounding given $v_t=v$, i.e. $\Z^*_t(\vec{s},v)$,  towards higher values, while maintaining the same expected marginal allocation $\Exv{v_t}{X^*_t(\vec{s},v_t)}$. This ensures that the state evolution probabilities will remain the same as the original randomized rounding, and hence this improved rounding algorithm can be coupled after time $t$ with the original randomized rounding (hence, will respect all the capacity constraints). 

This new rounding algorithm achieves strictly more expected total value, a contradiction to the optimality of the original randomized rounding. Now, if $\Y^*_t(\vec{s})=0$ let $\tau_t(\vec{s})=+\infty$ (so that the adaptive pricing does not pick the element in this situation). Otherwise, let $\tau_t(\vec{s})$ be the threshold at which $\Z^*_t(\vec{s},v)$ switched to zero, and let $p_t(\vec{s})= \Z^*_t(\vec{s},\tau_t(\vec{s}))$. With these choices, the adaptive pricing with randomized tie breaking simulates the conditional probabilities $\Z^*_t(\vec{s},v)$ and couples with the original randomized rounding of the optimal LP solution, so it will be an optimal online policy. \qed

\begin{remark}[{computing prices and tie-breaking probabilities in \cref{prop:adaptive-pricing}}]
 Given the optimal assignment of \ref{eq:lp-optimal}, prices and tie-breaking probabilities can easily be computed. At any time $t$ and feasible state $\vec{s}$,  find the minimum price $\tau_t(\vec{s})$ for which the probability that the value of the arriving element is above the price is at most $\frac{\Exv{v_t}{\X^*_t(\vec{s},v_t)}}{\Y^*_t(\vec{s})}$. After finding $\T_t(\vec{s})$, set the tie-breaking probability to $p_t(\vec{s})=\frac{1}{\Pr{v_t=\tau_t(\vec{s})}}\left(\frac{\Exv{v_t}{\X^*_t(\vec{s},v_t)}}{\Y^*_t(\vec{s})}-\Pr{v_t>\tau_t(\vec{s})}\right)$. An straightforward calculation shows this pricing policy has exactly the same state probabilities and expected social-welfare (LP objective) as that of the optimum LP solution. 
\end{remark}

\section{Beyond Constant Depth: the production constrained problem}
\newcommand{\KS}{{\mathcal{N}}}
\newcommand{\qedm}{\tag*{$\square$}}

\label{sec:production-simple}
In this section we formalize the production constrained Bayesian selection, a special case of laminar matroid Bayesian selection problem and then propose a PTAS for the optimal online policy for maximizing social-welfare. Similar to \Cref{sec:laminar}, to obtain our results we present an exponential-sized dynamic program that characterizes the optimum online policy and show how it can be written as a linear program. We then relax this linear program and explore how it can be rounded to a feasible online policy without considerable loss in expected social-welfare. In this section we leverage the special structure of this problem to go beyond constant depth.


\subsection{Problem description}

Consider a firm that produces multiple copies of $m$ different types of products over time. The firm offers these items in an online fashion to $n$ arriving unit-demand buyers. We assume each buyer $t=1,\ldots,n$ is only interested in one type of product and has a private value $v_t$ drawn independently from a known value distribution $F_t$ (which can be atomic or non-atomic).

 Buyers arrive at continuous times over the span of $T$ days, and reveal their value upon arrival.\footnote{Although the main goal of this paper is the selection problem and not the incentive compatible mechanism design, as we will mention later, all of our policies are pricing and hence truthful for myopic buyers.} \revcolor{We assume that the arrival time of buyers and the sequence of buyer types (i.e., which product type the buyer is interested in at each dau) are  known in advance. However, the values $v_1,\ldots,v_T$ unknown and  revealed sequentially to the decision maker.} At the end of $T$ days, the firm ships the sold items to the buyers.

Our goal is to find a feasible online policy for allocating the items to buyers to maximize \emph{social-welfare}, or equivalently, the sum of the valuations of all the served buyers. A \emph{feasible} policy should respect \emph{production constraints}, i.e., at any time the number of sold items of each type is no more than the number of produced items of that type. Moreover, it should respect the \emph{shipping constraint}, i.e. the total number of items sold does not exceed the shipping capacity $K$.\footnote{As a running example throughout the paper, the reader is encouraged to think of TESLA Inc. as the firm and its different models of electric cars, i.e. Model 3, Model X and Model S, as different product types.} 

  Suppose that at the beginning of day $i$,  the firm has produced $k^j_{i}$ items of type $j$. Let $B^j_i\subseteq[n]$ (referred to as \emph{bin}) denote the set of buyers of type $j$ arriving before day $(i+1)$ and $B^j\triangleq B^j_T$ denote the set of all buyers of type $j$. 
   The laminar structure corresponding to this problem can be seen in \Cref{fig:prod-laminar}. \revcolor{We assume that the structure of the laminar matroid in \Cref{fig:prod-laminar}, in particular, the shipping capacity $K$ and the production amounts $k_1^j,\ldots,k_T^j$ for each type of product $j$, are known in advance.}

\begin{figure}[htb]
  \centering
  \includegraphics[width=0.95\columnwidth]{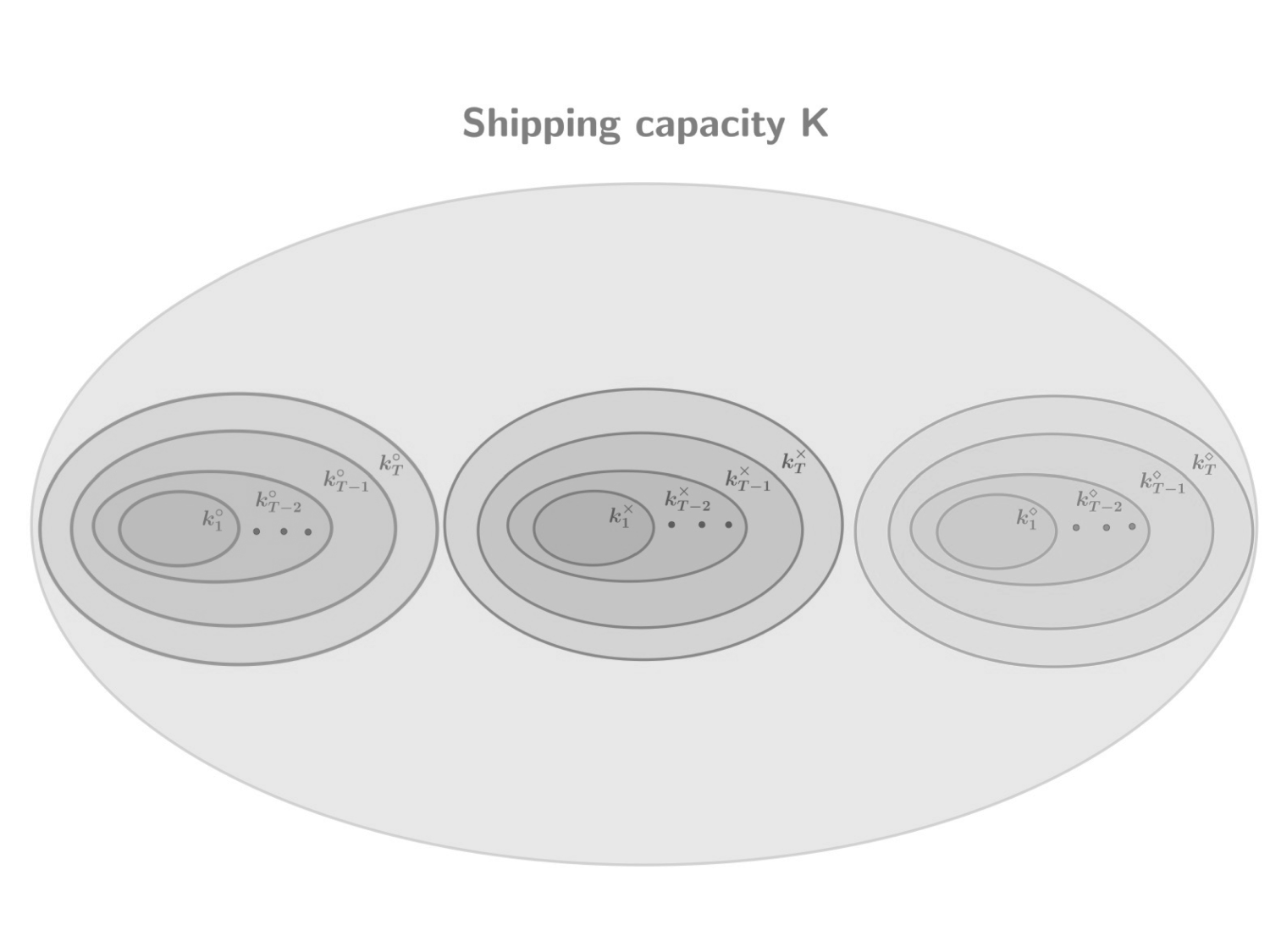}
     \caption{Production constrained Bayesian selection as a special case of Laminar Bayesian selection; each of the types $\circ$(red), $\times$(blue), and $\diamond$(green) has a corresponding collection of nested bins (path laminar), and these bins are inside an outer bin $[n]$ to model the shipping constraint.}
     \label{fig:prod-laminar}
\end{figure}

Similar to \Cref{sec:laminar}, we use the optimum online policy as a benchmark. Next we present a dynamic program that characterizes the optimum online policy.

\subsection{Linear programming formulation and expected relaxation}

Our production constrained Bayesian selection problem can be solved exactly using a simple exponential-sized dynamic program. In \Cref{sec:laminar}, we represented each state by keeping track of number of picked elements in each bin. In this section, because of the specific structure of our problem, we can simplify and represent each state by a vector $\vec{s}=[s_1,s_2,...,s_m]\in \mathbb{Z}^{m}$ maintaining the current number of sold products of each type. We say $\vec{s}$ is a \emph{feasible state} at time $t$ if it can be reached at time $t$ by a feasible online policy respecting all production constraints and the shipping constraint. It is possible to check whether $\vec{s}$ is feasible at time $t$ using a simple greedy algorithm.

Define $\V_t(\vec{s})$ to be the maximum total expected welfare that an online policy can obtain from time $t$ to time $n$ given $\vec{s}$. Define $\V_{t}(\vec{s}) = -\infty$ when $s$ is not feasible at time $t$ and $\V_{n+1}(\vec{s}) = 0$ for all $\vec{s}$.  Similar to \Cref{sec:lpformulation}, we can compute $\V_t(s)$ for the remaining values of $s$ and $t$ recursively using the following \emph{Bellman equation}:
\begin{equation}
\V_{t}(\vec{s})= \max_\tau \parens*{\ExX{v_t\sim F_t}{\parens*{v_t+\V_{t+1}(\vec{s}+\vec{e}_{j_t})}\cdot \1{v_t\geq \tau}}+\ExX{v_t\sim F_t}{\V_{t+1}(\vec{s})\cdot \1{v_t<\tau}}}.
\end{equation}
Note that the price $\tau_t(\vec{s})= \V_{t+1}(\vec{s})-\V_{t+1}(\vec{s}+\vec{e}_{j_t})$ maximizes the above equation, and so the final prices of an optimal online policy can be computed easily given the table values.  

As we mentioned is \cref{sec:laminar}, we can turn this dynamic program into a linear program almost exactly the same as (\ref{eq:lp-optimal}). We just replace $\vec{d}_t$ by $\vec{e}_{j_t}$ and redefine the set of forbidden neighboring states as
$$\partial S(t)\triangleq\left \{\vec{s}\in \mathbb{Z}^m: \left[\textrm{$\vec{s}$ is not a feasible state upon the arrival of buyer $t$}\right]~\&~\left[\exists j~\textrm{s.t.}~\vec{s}-\vec{e}_j\in \mathcal{S}\right]\right\}.$$
Where as before $\mathcal{S}\subset \mathbb{Z}^m$ is defined to be a finite set containing all possible feasible states at any time $t$. One can see that $\partial \mathcal{S}(t)$ has at most $O(n^{K+1})$ states.
 Unfortunately our linear programming formulation will still have exponential size. Nevertheless, without a shipping constraint, it can be solved in polynomial time. In fact, any online policy can be decomposed into $m$ separate online policies for type-specific \emph{sub-problems}; in each sub-problem, its corresponding policy only requires to respect the production constraints of its type. At the same time, the dynamic programming table of each sub-problem $j$ is polynomial-sized, as the state at time $t$ is essentially the number of sold products of type $j$ before $t$. Therefore the overall optimal online policy can be computed in polynomial time.

What if we relax the shipping constraint  to hold only in expectation (over the randomness of the policy/values)?  This relaxation is used in the prophet inequality literature and is termed as the \emph{expected relaxation}.

Next we formulate the  expected LP relaxation. First, re-define $\mathcal{S}\triangleq [1:{\max}_j~{k^j_T}]$ to be the set of possible states of each sub-problem.\footnote{Notably, we only need $\mathcal{S}$ to be a superset of all feasible states of each sub-problem $j$ at any time $t$.} Second, for each type $j$ and buyer $t\in B^j$, we use allocation variables $\X_t(s_j,v)$, marginal variables $\X_t(v)$, and state variables $\Y_t(s_j)$. These variable are defined as in (\ref{eq:lp-optimal}) and $s_j$ represents the number of items sold of type $j$ before the arrival of buyer $t$. \revcolor{We reiterate that variable $\X_t(s_j,v)$ represents the probability of being in state $s_j$ and having an allocation at time $t$ conditioned on $v_t=v$, $\X_t(v)$ represents the marginal probability of having an allocation at time $t$ conditioned on $v_t=v$, and $\Y_t(s_j)$ represents the probability of being at state $s_j$ at the beginning of time $t$.} We further use variables $\KS_j$ to represent the expected number of served buyers of each type $j$.
\begin{equation*}
\label{eq:lp-exante}\tag{LP$_3$}
\begin{array}{ll@{}ll}
\text{maximize}  & \displaystyle \sum_{t=1}^{n}\Exv{v_t}{v_t\cdot\X_t(v_t)} & &\\
\text{subject to}& \displaystyle\sum_{t\in B^j}\Exv{v_t}{\X_t(v_t)}\leq \KS_j,&~~~j=1,\ldots,m.\\
&\{\X_t(s_j,v),\X_t(v),\Y_t(s_j)\}_{t\in B^j}\in \Po^{\textrm{sub}_j},&~~~j=1,\ldots,m. \\
 &\displaystyle \sum_{j=1}^m \KS_j\leq K.&~~~\emph{(shipping capacity in expectation)}
\end{array}
\end{equation*}
where $B^j\subseteq [n]$ is the set of all buyers of type $j$ and $\Po^{\textrm{sub}_j}$ is the polytope of \emph{point-wise feasible online policies} for serving type $j$ buyers, defined by the following set of linear constraints:
\begin{equation*}
\begin{array}{ll@{}ll}
& \X_t(v)=\displaystyle\sum_{s_j\in \mathcal{S}}\X_{t}(s_j,v) &~~\forall v,~ t\in B^j,\\ \\
                                                         &0 \leq \X_{t}(s_j,v)\leq \Y_{t}(s_j)&~~\forall v,~s_j\in\mathcal{S}, t\in B^j, \\
                                                        &\Y_{t}(s_j)=\Y_{t'}(s_j)-\Exv{v_{t'}}{\X_{t'}(s_j, v_{t'})}+\Exv{v_{t'}}{\X_{t'}(s_j-1, v_{t'})}&~~\forall s_j\in \mathcal{S},~ t,t'\in B^j,~~~~~({\emph{\textrm{state update}}}) \\ 
                                                        &&~~[t'+1:t-1]\cap B^j=\emptyset \\ \\
                                                        
                                                        & \Y_{t_0}(0)=1 &~~t_0=\min\{t\in B^j\}\\ 
                                                           &\Y_t(s_j)=0&~~\forall s_j\in \partial \mathcal{S}^j(t), t\in B^j.~~~~({\emph{\textrm{feasibility check}}})
\end{array}
\end{equation*}
where $\partial \mathcal{S}(t)$ is the set of all forbidden neighboring states of the sub-problem of type $j$ at time $t$, i.e. 
$$\partial S^j(t)\triangleq \left\{s\in \mathbb{Z}: \left[\textrm{${s}$ is greater than the total production of type $j$ up to time $t$}\right]~\&~\left[s\leq \underset{j}{\max}~k_{T}^j+1\right]\right\}$$
Note that because of the collapse of the state space, \ref{eq:lp-exante} has polynomial size.

As every online policy for our problem induces a feasible online policy for each request type $j$, and because it respects the shipping capacity point-wise, we have the following proposition.
\begin{proposition}
\label{prop:lp-exante-relaxation}
\ref{eq:lp-exante} is a relaxation of the optimal online policy for maximizing expected social-welfare in the production constrained Bayesian selection problem.
\end{proposition}

\subsection{A Polynomial-Time Approximation Scheme (PTAS)}
Given parameter $\epsilon>0$, our proposed polynomial time approximation scheme  is based on solving a linear program with size polynomial in $n$ and an adaptive pricing mechanism with randomized tie breaking that rounds this LP solution to a $\left(1-O(\eps)\right)$-approximation. For notation purposes here and in \cref{sec:neg-dep-proof}, let $\{\X^*_t(\vec{s},v_t),\X^*_t(v_t),\Y^*_t(\vec{s})\}$ be the optimal solution of the linear program corresponding to the optimum online policy, and $\{\X^*_t(s_j,v_t),\X^*_t(v_t),\Y^*_t(s_j),\KS^*_j\}$ be the optimal assignment of \ref{eq:lp-exante} for the buyers $t\in B^j$.

\paragraph{Overview.} Consider the linear program of the optimal online policy and its expected relaxation (\ref{eq:lp-exante}). For a given constant $\delta>0$, we turn to one of these linear programs, depending on whether the shipping capacity $K$ is \emph{small} ($K\leq \frac{1}{\delta}$) or \emph{large} ($K> \frac{1}{\delta}$). In the former case, we pay the computational cost of solving the optimal policy LP and then round it exactly to a point-wise feasible online policy. In the latter case, we first reduce the large shipping capacity $K$ by a factor of $(1-\epsilon)$ to create some slack, and then solve \ref{eq:lp-exante} with this reduced shipping capacity (which has polynomial size). We then round the LP solution exactly by an adaptive pricing with randomized tie breaking policy. As we will show later, the resulting online policy respects all the constraints of the production constrained  Bayesian selection problem with high probability. 


\paragraph{The algorithm.} More precisely, we run the following algorithm (\cref{alg:ptas}).
\begin{algorithm}[htb]
\small
\algblock[Name]{Start}{End}
\algblockdefx[NAME]{START}{END}%
[2][Unknown]{Start #1(#2)}%
{Ending}
\algblockdefx[NAME]{}{OTHEREND}%
[1]{Until (#1)}
 \caption{ PTAS-Production-Constrained~($\epsilon,\delta$)
     \label{alg:ptas}}
      \begin{algorithmic}[1]
       \State{\textbf{Input}} parameters $\eps,\delta>0$.
        \If{$\left(K\leq \frac{1}{\delta}~\textrm{i.e. when shipping capacity is small}\right)$} 
        \State Solve the linear program of the optimal online policy.
        \State Given the optimal assignment, extract adaptive prices $\tau_t(\vec{s})$ and tie-breaking probabilities $\tau_t(\vec{s})~,\forall\vec{s}\in \mathcal{S}$.
        \Else {~$\left(K> \frac{1}{\delta}~\textrm{i.e. when shipping capacity is large}\right)$}
        \State Reduce the shipping capacity $K$ by a multiplicative factor $(1-\eps)$. 
        \State Solve the expected linear program (\ref{eq:lp-exante}) with the reduced shipping capacity of $K(1-\epsilon)$.
		\For{product types $j=1,2,\ldots,m$}
        \State  Given the optimal assignment corresponding to the variables of type $j$, extract adaptive prices $~~~~~~~~~~~~~~~~~$$\tau_t(s_j)$ and  tie-breaking probabilities $p_t(s_j)$ for every $s_j\in \mathcal{S}$ and every buyer $t$ such that $j_t=j$. 
        \EndFor
        \EndIf
        \State Offer the adaptive prices (with randomized tie breaking) computed above to arriving buyers based on their types $j$. In the exceptional cases when there is no remaining shipping capacity, offer the price of infinity.
      \end{algorithmic}
\end{algorithm}
\paragraph{Computing prices and tie-breaking probabilities.}  
Given $\{\X^*_t(\vec{s},v_t),\Y^*_t(\vec{s})\}$, the proof of \cref{prop:adaptive-pricing} in \cref{sec:proof-of-prop1} 
 gives a recipe to find $\tau_t(\vec{s})$ and $p_t(\vec{s})$ efficiently, so that the corresponding adaptive pricing with randomized tie-breaking policy  maintains the same expected marginal allocation $\Exv{v_t}{\X^*_t(\vec{s},v_t)}$ as the optimal online policy for every buyer $t$ and state $\vec{s}\in \mathcal{S}$, while having at least the same expected social-welfare. 

For the case of large shipping capacity, we apply exactly the same argument for each sub-problem $j$ separately. Given $\{\X^*_t(s_j,v_t),\Y^*_t(s_j),\KS^*_j\}$ for ${t\in B^j}$, we can efficiently find prices $\tau_t(s_j)$ and probabilities $p_t(s_j)$, so that the corresponding adaptive pricing with randomized tie-breaking for buyers with type $j$ maintains the same expected marginal allocation $\Exv{v_t}{\X^*_t(s_j,v_t)}$ for every $t\in B^j$ and $s_j\in \mathcal{S}$, while having at least the same expected social-welfare from serving each individual buyer of type $j$.\footnote{Note that when $K$ is large, we run separate pricing policies for each type $j$. Hence, given the type of buyer $t$, its offered price and tie-breaking probability is determined based on the current state of the sub-problem $j_t$, namely $s_{j_t}$.}

%
%

\paragraph{Feasibility, running time and social-welfare.} Clearly, \cref{alg:ptas} is a feasible online policy in the case of small shipping capacity (\cref{prop:adaptive-pricing}). In the case of large shipping capacity, as $\Y^*_t(s_j)=0$ for any forbidden neighboring state of sup-problem $j$, the same argument shows that it respects all of the production constraints of each type $j$. The policy also never violates the shipping capacity by construction, and hence is feasible. In terms of running time, if the shipping capacity is small, we solve solve the linear program of optimal online policy which has at most $n^{\frac{1}{\delta}}$ states, as no more than $\frac{1}{\delta}$ requests can be accepted. On the other hand, if the shipping capacity is large, we solve the expected linear program (\ref{eq:lp-exante}) which again can be solved in polynomial time. By setting $\delta=\eps^2/\log(1/\eps)$, \cref{alg:ptas} has running time $\texttt{poly}(n^{\frac{\log(1/\eps)}{\eps^2}})$. We further show that its expected welfare is at least $(1-O(\eps))$ fraction of the expected  welfare of the optimal online policy.

\begin{theorem}[PTAS for optimal online policy]
\label{THM:MAIN-1}
By setting $\delta=\frac{\eps^2}{\log(1/\eps)}$, \cref{alg:ptas} is a $(1-O(\eps))$-approximation for the expected social-welfare of the optimal online policy of the production constrained Bayesian selection problem, and runs in time $\texttt{poly}(n^{\frac{\log(1/\eps)}{\eps^2}})$.
\end{theorem}
\begin{corollary}[PTAS to maximize revenue] By applying Myerson's lemma from Bayesian mechanism design~\citep{hartline2012approximation,myerson1981optimal} and replacing each buyer value $v_t$ with her \emph{ironed virtual value} $\bar\phi_t(v_t)$, \cref{THM:MAIN-1} gives a PTAS for the optimal online policy for maximizing expected-revenue. 
\end{corollary}
\subsection{Analysis of the algorithm (proof of \texorpdfstring{\Cref{THM:MAIN-1}}{TEXT})}
\label{sec:analysis}
If the shipping capacity is small, i.e. $K\leq \frac{1}{\delta}$, \cref{alg:ptas} has the optimal expected social-welfare among all the feasible online policies, because of the same reason as the optimality of \ref{eq:lp-optimal} (\cref{prop:adaptive-pricing}). Next consider the expected LP in the case when $K>\frac{1}{\delta}$. By \cref{prop:lp-exante-relaxation}, its optimal solution is an upper bound on the social-welfare of any feasible online policy. By scaling the shipping capacity by a factor $(1-\eps)$, we change the optimal value of this LP by only a multiplicative factor of at least $(1-\eps)$. 

As sketched before, for each type $j$ the adaptive pricing policies $\{\tau_t(s_j),p_t(s_j)\}$ extract an expected value from buyer $t$ that is at least equal to the contribution of this buyer to the objective value of the expected LP. However, buyer $t$ can be served by the adaptive pricing policy of type $j_t$ only if the large shipping capacity has not been exceeded yet. So, to bound the loss, the only thing left to prove is that the probability of this bad event is small (as small as $O(\eps)$).


\revcolor{\paragraph{Concentration, negative dependency, and super-martingales.} In the case when the shipping capacity is large, let $X_t\in \{0,1\}$ be a Bernoulli random variable, indicating whether the resulting pricing policy of type $j_t$ serves the buyer $t$ or not. Note that  $\sum_{t}\Ex{X_t}\leq (1-\eps)K$, as \ref{eq:lp-exante}  ensures feasibility of the shipping constraint in expectation. Now, if the total count $\sum_{t}X_t$ concentrates around its expectation, we will be able to bound the probability of a bad event that the shipping capacity is exceeded.

Clearly, $\{X_t\}_{t\in B^1}, \{X_t\}_{t\in B^2},\ldots, \{X_t\}_{t\in B^m}$ are mutually independent, as we run a separate adaptive pricing policy for each type $j$. However, the indicator random variables of the same type are not mutually independent with each other. So, for proving the required concentration, a sub-Gaussian concentration bound, such as the Chernoff bound or the Azuma bound, cannot be applied immediately. However, we can still hope to obtain our desired concentration bound if the sequence $\{{Y}_t\triangleq\sum_{\tau\leq t}(X_\tau-\Ex{X_\tau})\}$ forms a super-martingale. We first prove the following technical lemma. 

\begin{lemma}[Super-martingale property] Let $X_1,X_2,\ldots$ be a sequence of Bernoulli random variables such that the sequence $\{{Y}_t\triangleq\sum_{\tau\leq t}(X_\tau-\Ex{X_\tau})\}$ is a super-martingale, that is, $\forall t:\mathbb{E}\left[Y_t|Y_{t-1}\right]\leq Y_{t-1}$. Then for $\epsilon\in[0,\frac{1}{2}]$:
$$
\Pr{Y_t>\epsilon\cdot\Ex{\sum_{\tau\leq t}X_{\tau}}}\leq e^{-2\epsilon^2(\Ex{\sum_{\tau\leq t}X_{\tau}})}
$$
    \label{lemma:multiplicative-azuma}
\end{lemma}
\begin{proof}{Proof.}
Let $\Pi_t\triangleq \prod_{\tau\leq t}(1+\epsilon)^{X_\tau}(1-\epsilon)^{\Ex{X_\tau}}=(1+\epsilon)^{\sum_{\tau\leq t}X_\tau}(1-\epsilon)^{\sum_{\tau\leq t}\Ex{X_\tau}}$ and $\Pi_0\triangleq 1$. Note that:
$$
\Pi_t=\Pi_{t-1}(1+\epsilon)^{X_t}(1-\epsilon)^{\Ex{X_t}}\leq \Pi_{t-1}(1+\epsilon X_t-\epsilon\Ex{X_t}),
$$
where the last inequality holds because $(1+\epsilon)^x(1-\epsilon)^y\leq (1+\epsilon(x-y))$ as long as $\lvert x-y\rvert<1$. Therefore, by taking conditional expectation from both sides, we have:
\begin{equation}
\label{eq:martingale}
\mathbb{E}\left[\Pi_t|\Pi_{t-1}\right]\leq \Pi_{t-1}+\epsilon\mathbb{E}\left[X_t-\Ex{X_t}|\Pi_{t-1}\right]\leq \Pi_{t-1},
\end{equation}
where in the last inequality we use the fact that $Y_t$ is a super-martingale and hence:
$$
\mathbb{E}\left[X_t-\Ex{X_t}|\Pi_{t-1}\right]=\mathbb{E}\left[Y_t-Y_{t-1}|\Pi_{t-1}\right]=\mathbb{E}\left[Y_t-Y_{t-1}|Y_{t-1}\right]\leq 0.
$$
Applying the inequality in \eqref{eq:martingale} recursively and taking expectations, we have:
$$\Ex{\Pi_t}\leq \Ex{\Pi_{t-1}}\leq \ldots\leq \Ex{\Pi_0}=1~.$$
Now, let $\alpha=\Ex{\sum_{\tau\leq t}X_\tau}$.  Using the Markov inequality, for any $\delta\in[0,1]$, we have:
$$
\Pr{\ln(\Pi_t)>\delta \alpha}=\Pr{\Pi_t>e^{\delta \alpha}}\leq \Ex{\Pi_t}e^{-\delta \alpha}\leq e^{-\delta \alpha}
$$
At the same time, note that $\ln(\Pi_t)=\ln(1+\epsilon)(Y_t+\alpha)-\alpha\ln(\frac{1}{1-\epsilon})$. We then have the following:
$$
\ln(\Pi_t)>\delta \alpha \Longrightarrow \frac{\ln(1+\epsilon)}{\ln(\frac{1}{1-\epsilon})}(Y_t+\alpha)-\alpha>\frac{\delta\alpha}{\ln(\frac{1}{1-\epsilon})}\Longrightarrow \frac{\epsilon}{\frac{\epsilon}{1-\epsilon}-\frac{\epsilon^2}{2(1-\epsilon)^2}}(Y_t+\alpha)-\alpha>\frac{\delta (1-\epsilon)}{\epsilon}\alpha,
$$
where in the last inequality we used the facts that for every $\epsilon\in[0,1]$, we have:
\begin{equation}
    \ln(1+\epsilon)\leq \epsilon~~\textrm{and}~~~\frac{\epsilon}{1-\epsilon}-\frac{1}{2}\left(\frac{\epsilon}{1-\epsilon}\right)^2\leq \ln\left(1+\frac{\epsilon}{1-\epsilon}\right)\leq \frac{\epsilon}{1-\epsilon}
\end{equation}
By rearranging the terms, we conclude that:
$$
\Pr{\frac{1}{1-\frac{\epsilon}{2(1-\epsilon)}}(Y_t+\alpha)-\alpha>\frac{\delta}{\epsilon}\cdot \alpha}<e^{-\delta \alpha}\Longrightarrow \Pr{Y_t+\alpha-\alpha(1-\frac{\epsilon}{2(1-\epsilon)})>\frac{\delta}{\epsilon}\cdot \alpha}<e^{-\delta K}
$$
By setting $\delta=2\epsilon^2$, using the fact that $\epsilon\in[0,\frac{1}{2}]$, and rearranging the terms, we have:
$$
\Pr{Y_t>\epsilon\cdot\alpha}\leq \Pr{Y_t>-\alpha\frac{\epsilon}{2(1-\epsilon)})+2\epsilon\cdot \alpha}<e^{-2\epsilon^2 \alpha}~,
$$
which finishes the proof of the lemma.
\qed
\end{proof}



We now prove the following proposition, assuming the required super-martingale property for the allocations of each type. 
\begin{proposition} 
\label{prop:social-welfare-1}
If for every request type $j$, the random variables $\{X_t\}_{t\in B^j}$ satisfy the super-martingale property stated in \Cref{lemma:multiplicative-azuma} and if $K>\frac{1}{\delta}=\frac{\log(1/\epsilon)}{\epsilon^2}$ for some $\epsilon\in[0,\frac{1}{2}]$, then the probability that the shipping capacity $K$ is exhausted is $O(\eps)$. 
\end{proposition}
\begin{proof}{Proof.}
First of all, if for every type $j$ the super-martingale property in \Cref{lemma:multiplicative-azuma} holds for all Bernoulli random variables $\{X_t\}_{t\in B^j}$, then because of the mutual independence of the Bernoulli variables corresponding to allocations of different types, we can say that all Bernoulli random variables $\{X_t\}$ satisfy the super-martingale property stated in \Cref{lemma:multiplicative-azuma}. 




Note that $\sum_{t}\Ex{X_t}= (1-\eps)K$, simply because for any optimal solution of the linear program \eqref{eq:lp-exante}, the shipping capacity (in expectation) is clearly binding (otherwise, we can slightly increase the allocation probability of some element and it will contradict the optimality of the LP solution). Therefore we have:
\begin{align*}
\Pr{\displaystyle\sum_{t}X_t>K}= \Pr{\displaystyle\sum_{t}(X_t-\Ex{X_t})>\eps K}\leq e^{-{2K(1-\epsilon)\frac{\epsilon^2}{(1-\epsilon)^2}}}=O(\epsilon)~,
\end{align*}
where the first inequality holds because of \Cref{lemma:multiplicative-azuma}, and the last equality holds
by setting $K=\frac{\log(1/\epsilon)}{\epsilon^2}$. This will finish the proof. \qed
\end{proof}}

\subsubsection{Negative dependency for optimal online policy}
\label{sec:neg-dep-proof}
\revcolor{Fix a product type $j$. For notation simplicity, re-index $\{X_t\}_{t\in B^j}$ as $X_{t_1},\ldots,X_{t_l}$, where $t_1\leq t_2\leq\ldots\leq  t_l$ and $l\triangleq\lvert B^j\rvert$. To show that the sequence of random variables $\{X_{t_i}\}$ satisfies the super-martingale property described in \Cref{lemma:multiplicative-azuma}, 
one needs to show that for the adaptive pricing with randomized tie-breaking used in \cref{alg:ptas}, the probability of accepting a buy request at time $t$ can only decrease conditioned on more requests being accepted in the past. More precisely, we present and prove the following proposition.
\begin{proposition}\label{ncdthm}
	Let $X_t \in \{0 ,1\}$ be a random variable indicating whether the policy serves buyer $t$ or not. Then the variables $X_1, \ldots, X_n$  satsify the super-martingale property in \Cref{lemma:multiplicative-azuma}, that is, the sequence ${Y}_t\triangleq\sum_{\tau\leq t}(X_\tau-\Ex{X_\tau})$ for $t=1,\ldots,n$ is a super-martingale, or equivalently,  $\forall t\in[n]:\mathbb{E}\left[Y_t|Y_{t-1}\right]\leq Y_{t-1}$.
\end{proposition}}

\begin{proof}{Proof.}

First observe that given $\KS^*_j$ as the optimal soft shipping capacity that needs to hold only in expectation, $\{\X^*_t(s_j,v_t),\X^*_t(v_t),\Y^*_t(s_j)\}$ for $t\in B^j$ is indeed the optimal solution of the following linear program.

\begin{equation*}
\label{eq:lp-opt-sub}\tag{LP-sub$_j$}
\begin{array}{ll@{}ll}
\text{maximize}  & \displaystyle \sum_{t\in B^j}\Exv{v_t}{v_t\cdot\X_t(v_t)} & &\\
\text{subject to}& \displaystyle\sum_{t\in B^j}\Exv{v_t}{\X_t(v_t)}= \KS^*_j,&~~~\emph{(soft shipping constraint)}\\
&\{\X_t(s_j,v),\X_t(v),\Y_t(s_j)\}_{t\in B^j}\in \Po^{\textrm{sub}_j}.&
\end{array}
\end{equation*}
Therefore, it is enough to show the same property holds for another adaptive pricing with randomized tie-breaking algorithm that is used for exactly rounding \ref{eq:lp-opt-sub}, as both of these rounding algorithms have the same allocation distribution for the buyers of type $j$. 

For simplicity of the proofs in this section, we assume that the valuations are non-atomic.\footnote{For the case of atomic distributions, one can think of dispersing each value distribution first to get non-atomic distributions, and then proving our desired super-martingale property for any small dispersion. Then the same property for the original atomic distribution can be deduced from the super-martingale property of the dispersed distribution for small enough dispersion.} Note that under this assumption, there will be no need for randomized tie breaking, and indeed our rounding algorithm will be pure adaptive pricing. Now we prove our claim in two steps. 
\begin{enumerate}
\item \textbf{Step 1:} by using LP duality, we show that the optimal online policy of the sub-problem of type $j$ with an extra soft shipping constraint $\KS^*_j$ is indeed the optimal online policy for an instance that has no soft constraint and all the values are \emph{shifted} by some number $\lambda^*$, i.e. $\hat{v}_t=v_t-\lambda^*$.
\item \textbf{Step 2:} we show that the optimal online policy of a particular sub-problem $j$, whether value distributions have negative points in their support or not, \revcolor{satisfies the super-martingale property described in \Cref{lemma:multiplicative-azuma} (or equivalently, the probability that a new arriving request gets accepted decreases as more elements are accepted in the past).}
\end{enumerate}
\revcolor{Putting the two pieces together, we prove our desired super-martingale property (as in \Cref{lemma:multiplicative-azuma}) among $X_{t_1},...,X_{t_l}$ as desired. In the remaining of this section, we prove these two steps.}

\begin{proof}{Proof of Step 1.}
Consider (\ref{eq:lp-opt-sub}) that captures the optimal online policy of sub-problem $j$. We start by moving the soft shipping constraint into the objective of of  (\ref{eq:lp-opt-sub}) and writing the Lagrangian.
\begin{equation*}
\begin{array}{ll@{}ll}
\underset{\lambda}{\text{minimize}}~\underset{\X}{\text{maximize}}  & \displaystyle \mathcal{L}(\X, \lambda)=\sum_{t\in B^j}\Exv{v_t}{v_t\cdot\X_t(v_t)}+\lambda(\KS_j^*-\sum_{t\in B^j}\Exv{v_t}{\X_t(v_t)})&\\
&\{\X_t(s_j,v),\X_t(v),\Y_t(s_j)\}_{t\in B^j}\in \Po^{\textrm{sub}_j}.&
\end{array}
\end{equation*}
Let $\lambda^*$ be the optimum dual solution. By dropping the constant terms and rearranging we get the following equivalent program for the optimal solution:
\begin{equation*}
\begin{array}{ll@{}ll}
{\text{maximize}}  & \displaystyle \sum_{t\in B^j}\Exv{v_t}{(v_t-\lambda^*)\cdot\X_t(v_t)}\\
&\{\X_t(s_j,v),\X_t(v),\Y_t(s_j)\}_{t\in B^j}\in \Po^{\textrm{sub}_j}.&
\end{array}
\end{equation*}
This shows that the optimal online policy respecting the soft shipping  constraint $\KS_j^*$ is equivalent to the optimal online policy for an instance of the problem where all the values are shifted by some constant $\lambda^*$. \qed
\end{proof}

\begin{proof}{Proof of Step 2.}	We only need to show that the super-martingale property holds for the dynamic programming that solves each sub-problem, as the distributions are non-atomic and there is a unique deterministic optimal online policy, characterized both by the LP and the dynamic programming. Consider sub-problem $j$. We use induction to show by serving more customers in the past, the prices for new buyers increase. Let $s_t^j$ denote the total number of products of type $j$ that have been sold up to the arrival of buyer $t$. Note that the algorithm only needs $s_t^j$ to decide whether buyer $t$ should be served.

Let $D_t(s_t^j)$ denote the maximum total expected welfare that an online policy can obtain from time $t$ to $n$, assuming that it starts from state $s_t^j$. Also let $C_t^j$ denote the set of production checkpoints of type $j$ that occur at or after time $t$.
	Using the Bellman equations we have	
\[   
D_{t-1}(s_{t-1}^j) = 
   \left\{
     \begin{array}{@{}l@{\thinspace}l}
       D_t(s_{t-1}^j)    &~~\textrm{if}~~ \min\left(\left\{k_i^j-s_{t-1}^j\right\}_{i\in C_{t-1}^j}\right)=0\\
       \Exv{v_{t-1}}{\max(D_t(s_{t-1}^j+1)+v_{t-1}, D_t(s_{t-1}^j))} &~~\textrm{if}~~ \min\left(\left\{k_i^j-s_{t-1}^j\right\}_{i\in C_{t-1}^j}\right)>0 \\
     \end{array}
   \right.
\]

As the base of our induction, we know that if we serve the last buyer, the probability that we serve any other buyers does not increase. Now assume while serving buyer $t$, we have
\begin{align}\label{inductionHypo}
	D_t(s_t^j)-D_t(s_t^j+1)\leq D_t(s_t^j+1)-D_t(s_t^j+2).
\end{align}
Note that this shows the price offered to buyer $t$ increases if we serve more buyers before buyer $t$.
When buyer $t-1$ arrives, we need to show
\begin{align}\label{inductionStep}
	D_{t-1}(s_{t-1}^j)-D_{t-1}(s_{t-1}^j+1)\leq D_{t-1}(s_{t-1}^j+1)-D_{t-1}(s_{t-1}^j+2)
\end{align}
which is equivalent to
$$D_{t-1}(s_{t-1}^j)+D_{t-1}(s_{t-1}^j+2)\leq 2D_{t-1}(s_{t-1}^j+1).$$
Note that this property is linear in the terms involved. So it is enough to assume that the value $v_{t-1}$ is deterministic first and prove the above inequality. Then by linearity of expectation, the inequality would hold in the general case.

Note that for the case where $\min\left(\left\{k_i^j-s_{t-1}^j\right\}_{i\in C_{t-1}^j}\right)<2$, the inequality holds trivially because we assume $D_{t-1}(s)=-\infty$ for any non-negative integer $s$ such that $\min\left(\left\{k_i^j-s\right\}_{i\in C_{t-1}^j}\right)<0$.
According to our induction hypothesis, if $D_{t-1}(s_{t-1}^j+2)$ is updated, then the other two variables are updated as well. More precisely, if $D_{t-1}(s_{t-1}^j+2)=D_t(s_{t-1}^j+3)+v_{t-1}$, then
\begin{align*}
	&D_{t-1}(s_{t-1}^j+1)=D_t(s_{t-1}^j+2)+v_{t-1},\\
	&D_{t-1}(s_{t-1}^j)=D_t(s_{t-1}^j+1)+v_{t-1}.
\end{align*}
In a similar way, if $D_{t-1}(s_{t-1}^j+1)=D_t(s_{t-1}^j+2)+v_{t-1}$, then
\begin{align*}
	D_{t-1}(s_{t-1}^j)=D_t(s_{t-1}^j+1)+v_{t-1}.
\end{align*}

Considering these relations, we have three different cases. (i) none of these variables are updated. In this case, \cref{inductionStep} turns into \cref{inductionHypo} which holds according to our induction hypothesis. (ii) all of these variables are updated. In this case, \cref{inductionStep} turns into
\begin{align*}
	D_t(s_t^j+1)+v_{t-1}-D_t(s_t^j+2) - v_{t-1} \leq D_t(s_t^j+2) + v_{t-1} - D_t(s_t^j+3) - v_{t-1}
\end{align*}
which holds again according to the induction hypothesis.
Finally, (iii) the case where $D_{t-1}(s_{t-1}^j)$ is updated and $D_{t-1}(s_{t-1}^j+2)$ is not updated. In this case, we can write \cref{inductionStep} as
\begin{align*}
	D_t(s_t^j+1)+v_{t-1}+D_t(s_t^j+2)\leq 2 \max(D_t(s_t^j+2)+v_{t-1}, D_t(s_t^j+1))
\end{align*}
and this always holds because for any two values $x, y$, $\frac{x+y}{2}\leq \max(x, y)$. \qed
\end{proof}
\end{proof}

\subsection{Discussion: extension to laminar over non-constant depth sub-problems.}
As we saw earlier in \Cref{sec:laminar-mark}, we can prove the correctness of our polynomial time approximation scheme only if the depth of the laminar family is constant. At the same time, as an observation, if one thinks of the production constrained Bayesian selection problem as an instance of the laminar Bayesian selection, then the depth of the corresponding laminar family will simply be $T+1$, and hence is not a constant. Yet, as we saw in \cref{sec:production-simple}, our approach in that section could yield to a PTAS. Can we still see this PTAS as a special case of our PTAS for the constant-depth laminar?

This discrepancy can easily be explained by extending our result for the constant depth laminar Bayesian selection to a generalization where every element is replaced by a \emph{sub-problem}. With this view, in the production constrained Bayesian selection we indeed have only a 1-level tree (connecting the root to the type-specific sub-problems), and each leaf of this 1-level tree plays the role of one of the type-specific sub-problems. To make the analysis work, we require two things from each sub-problem. First, a local optimum online policy for each sub-problem, for possibly atomic or even non-positive value distributions, should be computable in polynomial time. Second, the selection rule of this local optimum online policy should satisfy the required negative dependency (i.e., that the probability of accepting an element decreases as more elements are being accepted in the past, which we also referred to as the super-martingale property in \Cref{lemma:multiplicative-azuma}). Having these two properties, the same proof as in this section can be used, as by replacing each element with a group of negatively dependent elements (in the same sense as in \Cref{lemma:multiplicative-azuma}) we still have the required concentration. The proof of this generalization is basically the same, so we omit this proof to avoid redundancy.

\section{Conclusion}
\label{sec:conclusion}
In this paper we took the first stab at designing polynomial time approximation schemes for Bayesian online selection problems. In this model, the goal is to serve a subset of arriving customers in a way that maximizes the expected social welfare while respecting certain capacity or structural constraints. We presented two polynomial time approximation schemes when the set of allowable customers is restricted either by a laminar family with constant depth or by joint production/shipping constraints. Our algorithms are based on rounding the solution of a hierarchy of linear programming relaxations that approximate the optimum solution within any degrees of accuracy. We hope that benchmarks similar to the type of linear programming hierarchy that we proposed here can lead to more insights as well as new and interesting algorithms for this class of stochastic online optimization problems (or even beyond).

%


%
 \begin{APPENDIX}{}
 \section{Re-inventing the Wheel: Prophet Inequalities Using LP}
 \label{sec:lp-prophet}
The benefits of our linear programming approach for the online Bayesian selection problem are two-fold. So far, we have seen our LPs give us a systematic way of describing the optimum online benchmark and its relaxations, which can be easily generalized to other combinatorial domains (e.g. matroids). In this section, we show some further applications of this approach in the classic single-item prophet inequality problem (defined formally in \Cref{sec:prophet}).  We show how to use this LP to design approximate pricing mechanisms with respect to the optimal offline in a \emph{modular way}, and therefore re-deriving simpler proofs for a couple of already existing prophet inequalities. We believe this approach can be useful for other settings as well, which we leave as future research directions. For the ease of exposition, we focus on non-atomic distributions. The case of general distributions can be easily handled by adding randomized tie-breaking to our mechanisms in a straightforward fashion.\footnote{\revcolor{A preliminary conference version of some of the results in this appendix section had appeared in \cite{niazadeh2018prophet}; in this section we expand on those results and provide all the technical details.}}

\subsection{Classic single-item Prophet inequality problem}
\label{sec:prophet}
In single item prophet inequality problem, a seller is interested in selling an item to a sequence of $n$  arriving buyers. Each buyer $i$ has a value $v_i$ for the item. This value is independently drawn from a distribution $\Fi{i}$. Buyers arrive one by one and reveal their values. Upon the arrival of a buyer, the seller decides whether to sell the item or move on to the next buyer. The goal is to maximize the expected value of the selected buyer. We consider the setting where the sequence of distributions $\Fi{1},\Fi{2},\ldots,\Fi{n}$ is picked by an oblivious adversary up front. We assume the seller knows the distributions in advance but does not know the order in which the buyers arrive  (an important distinction with previous sections of our paper). 

In contrast to previous sections, in which the goal was to design policies that approximate the optimum online benchmark, here we focus on obtaining prophet inequalities, where the goal is to design online policies that are competitive with respect to the optimum offline (or omniscient prophet) benchmark. In the single-item problem, this benchmark is simply equal to $\Ex{\underset{i\in[n]}{\max}~{v_i}}$.

\subsection{LP characterization of the optimum online benchmark} 

Let $[n]$ be a sequence of buyers arriving over times. Without loss of generality assume that at each time $t=1,\ldots,n$ buyer $t$ arrives. Let $\opton([n])$ be the optimum online mechanism given the sequence of arriving buyers $[n]$. We seek to find a linear programming characterization for $\opton([n])$. Note that using a simple backward induction one can find such an optimum policy; However, the introduced LP sheds more insight on the structure of this policy and helps us with designing approximate policies with respect to the optimum offline (i.e. prophet inequalities).

To write down the linear program for the single item prophet inequality problem, similar to (\ref{eq:lp-optimal}), let $\x{t}(v)$ denote the probability that the policy allocates the item to buyer $t$ conditioned on the event that the value of this buyer is equal to $v$. Our linear programming has variables $\x{t}(v)$ for every $t\in [n]$ and every $v\in \textrm{supp}(\Fi{t})$, where $\textrm{supp}(.)$ denotes the support of its input distribution. \footnote{\revcolor{In the case of non-atomic distributions, this LP is essentially a continuous program with uncountably many variables. In the case of discrete distributions, the LP has countably many variables, and if the support is bounded the LP has finitely many variables.}} We then try to impose constraints on these variables to guarantee that the solution of the LP is  online implementable, without losing anything in the expected allocated value. Formally speaking, consider the following linear program, which we denote by $\lpon([n])$:
\begin{align*}
\raisetag{85pt}
\label{eq:lp-opt-on}
\begin{array}{ll@{}ll}
\text{maximize}  & \displaystyle \sum_{t=1}^{n}\Exv{v_{t}\sim\Fi{t}}{v_{t}\cdot \x{t}(v_{t})}  & &\tag{\texttt{LP-ONLINE}($[n]$)}\\
\text{subject to}& \displaystyle \x{t}(v)\leq 1-\sum_{t'<t}\Exv{v_{t'}\sim\Fi{t'}}{\x{t'}(v_{t'})} ,&~~~~~~~\forall v\in \textrm{supp}(\Fi{t}),~t=2,\ldots,n\\
& \x{1}(v)\leq 1, &~~~~~~~\forall v\in \textrm{supp}(\Fi{1})\\
& \x{t}(v)\geq 0, &~~~~~~~\forall v\in \textrm{supp}(\Fi{t}),~t=1,\ldots,n
\end{array}
\end{align*}

It is not hard to see that every feasible online policy induces a feasible solution for \ref{eq:lp-opt-on}, by setting $\x{t}(v)$ to be the allocation probabilities of this policy. In fact, no allocation happens at time $t$ if the item has been allocated at some time $t'<t$. Therefore, by taking an expectation with respect to the buyer values arriving at times $t'=1,\ldots, t-1$, $\x{t}(v)$ will be at most equal to $1-\sum_{t'<t}\Exv{v_{t'}\sim\Fi{t'}}{\x{t'}(v_{t'})}$. More interestingly, the converse is also true. While we can prove the converse by applying \Cref{prop:adaptive-pricing} and slightly modifying the LP, we can also prove it directly. In order to be self-contained, we provide the direct proof here.

\begin{proposition}
\label{prop:rounding-exact}
Given any feasible assignment $\{\x{t}(v)\}$ for \ref{eq:lp-opt-on}, there exists a feasible online policy with an expected allocated value equal to the objective value of the LP under $\{\x{t}(v)\}$.
\end{proposition}
\begin{proof}{Proof.} Define $q_t\triangleq 1-\sum_{t'<t}\Exv{v_{t'}\sim\Fi{t'}}{\x{t'}(v_{t'})}$ for $t\geq 2$, and $q_1\triangleq 1$. Consider the following randomized rounding policy: at time $t\geq 1$, if the item has already been allocated do nothing. If it has not yet been allocated,  upon realizing the value $v_{t}$ flip an independent coin with heads probability  of $\frac{\x{t}(v_{t})}{q_t}$. Now, if the coin flips heads allocate the item and terminate. Otherwise, continue to the next buyer. 

Clearly the above policy is online and feasible, i.e. it sells the item to only one buyer. To compare its expected allocated value with the objective value of the LP under the assignment $\{\x{t}(v)\}$, we first claim that $q_t$ is equal to the probability that this randomized policy reaches time $t$, i.e. with probability $q_t$ the policy does not sell the item to any buyer arriving before time $t$. We prove this claim by induction. Clearly $q_1=1$ satisfies this property. As the induction hypothesis, suppose the policy reaches time $t\geq 2$ with probability $q_t$. To prove the inductive step, we have:
\begin{align*}
\Pr{\textrm{reaching time $t+1$}}&=\Pr{\left( \textrm{reaching time $t$} \right)\&\left(\textrm{no allocation at time $t$}\right)}\\
&=\Pr{\textrm{no allocation at time $t$}|\textrm{reaching time $t$}}\cdot q_t\\
&=\Exv{v_{t}\sim\Fi{t}}{\Pr{\textrm{no allocation at time $t$}|\left(\textrm{reaching time $t$}\right)~\&~v_{t}}}\cdot q_t \\
&=\Exv{v_{t}\sim\Fi{t}}{1-\frac{\x{t}(v_{t})}{q_t}}\cdot q_t=q_t-\Exv{v_{t}\sim\Fi{t}}{\x{t}(v_{t})}=q_{t+1}
\end{align*}
Next we claim that by conditioning the realized value at time $t$ to be $v$, the policy allocates the item with probability $\x{t}(v)$. This is simply true because the policy reaches time $t$ with probability $q_t$, and then conditioned on reaching time $t$ and realizing value $v$ allocates the item with probability $\frac{\x{t}(v)}{q_t}$. Finally, as $\{\x{t}(v)\}$ are the allocation probabilities of the policy (as we just proved), the expected allocated value at time $t$ is equal to $\Exv{v_{t}\sim\Fi{t}}{v_{t}\cdot \x{t}(v_{t})}$. The proof of the proposition is then finished by summing over all $t$.
\qed
\end{proof}

\subsection{Relaxation and rounding}
The goal of this section is to propose two sequential pricing policies, one for the case of non-identical distributions and one for the case of identical distributions, so that they obtain $\tfrac{1}{2}$ and $1-\tfrac{1}{e}$ fractions of the expected value of the omniscient prophet benchmark, respectively. To this end, we use \ref{eq:lp-opt-on} and the rounding algorithm proposed in \Cref{prop:rounding-exact}, and in a modular fashion design new algorithms satisfying the classic prophet inequality of \cite{krengel1978semiamarts} and the semi-optimal prophet inequality of \cite{correa2017posted} and \cite{esfandiari2017prophet}.

Our approach is based on the expected relaxation of \ref{eq:lp-opt-on}. Suppose the seller intends to sell the item, but rather than selling the item to only one buyer for every profile of buyer values, it has a relaxed constraint of selling the item to one person in expectation over buyer values. In the expected relaxation benchmark the seller only needs to sell the item to each buyer $t$ with probability $q_t$, where $\sum_t q_t\leq 1$. Clearly, the maximum expected value of this relaxation is an upper-bound on the expected value of the optimum offline mechanisms, as the omniscient prophet allocates the item to only one buyer point-wise. Moreover, the 
following linear program captures the expected relaxation:
\begin{align*}
\raisetag{70pt}
\label{eq:lp-ex-ante}
\begin{array}{ll@{}ll}
\text{maximize}  & \displaystyle \sum_{t=1}^{n}\Exv{v_{t}\sim\Fi{t}}{v_{t}\cdot \x{t}(v_{t})}  & &\tag{\texttt{Expected-LP}}\\
\text{subject to}& \displaystyle \sum_{t\in[n]}\Exv{v_{t}\sim\Fi{t}}{\x{t}(v_{t})} \leq 1,& \\
& \x{t}(v)\geq 0, &~~~~~~~\forall v\in \textrm{supp}(\Fi{t}),~t=1,\ldots,n
\end{array}
\end{align*}
Fix a feasible assignment $\{\x{t}(v)\}$ for the above LP, and let $q_t=\Exv{v_{t}\sim\Fi{t}}{\x{t}(v_{t})}$. Note that $\sum_t q_t\leq 1$. Now, one can replace $\{\x{t}(v)\}$ with the following assignment, which obtains at least as much expected value as before and is a feasible assignment for \ref{eq:lp-ex-ante}:
\begin{equation*}
\x{t}'(v)=
  \begin{cases}
       1 &\quad\quad\quad\quad   v\geq T_t(q_t),\\
       0 &\quad\quad\quad\quad\quad  \textrm{o.w.}
     \end{cases}
\end{equation*}

where $T_t(q_t)$ is the price corresponding to the quantile $q_t$ of the distribution $\Fi{t}$. More precisely, we define $T(q)\triangleq \underset{p\in \mathbb{R}}{\textrm{argmin}}~\mathbb{P}_{v\sim\mathcal{F}}\left[{v\geq p}\right]=q$. Under the pricing allocations $\{\x{t}'(v)\}$, the objective value of the expected LP is equal to $\sum_{t=1}^n V_t(q_t)$, where $V_t(q_t)$ is the concave value-curve of the distribution $\Fi{t}$. More precisely, we define $V(q)\triangleq \Exv{v\sim\mathcal{F}}{v\cdot \mathds{1}\{v \geq T(q)\}}$. By putting all the pieces together, the optimal solution to \ref{eq:lp-ex-ante} is $\x{t}^*(v)=\mathds{1}\{v\geq T_t(q^*_t)\}$, where $\mathbf{q}^*$ is the optimal solution of the following convex program:
\begin{align*}
\raisetag{70pt}
\label{eq:con-ex-ante}
\begin{array}{ll@{}ll}
\text{maximize}  & \displaystyle \sum_{t=1}^{n}V_t(q_t)& &\tag{\texttt{Expected-Conv}}\\
\text{subject to}& \displaystyle \sum_{t\in[n]}q_t\leq 1~,~q_t\geq 0&~~~~~~~t=1,\ldots,n
\end{array}
\end{align*}

\begin{remark} Note that the optimal solution of the expected LP relaxation (\ref{eq:lp-ex-ante}) can be computed by only knowing the set of distributions $\{\Fi{t}\}_{t\in [n]}$, and without the need to know the ordering in which the buyers arrive. In other words, this benchmark, similar to optimum offline, is \emph{order oblivious}; no matter what the ordering of the arriving buyers is, the expected LP relaxation yields the same solution. 
\end{remark}

\paragraph{Rounding for the non-identical distributions.}
We now start with the optimal solution of the expected LP relaxation described above, i.e. $\{\x{t}^*(v)\}$, and modify it so that it becomes online implementable. In a nutshell, consider $\x{t}(v)=\tfrac{1}{2}\x{t}^*(v)$. We show this solution is feasible for \ref{eq:lp-opt-on}.
\begin{proposition}The expected value obtained by Algorithm~\ref{alg:non-identical} is at least $\tfrac{1}{2}$ of optimum offline.
\end{proposition} 
\begin{proof}{Proof.}
Suppose $\x{t}^*(v)=\mathds{1}\{v\geq T_t(q^*_t)\}$ is the optimal assignment of (\ref{eq:lp-ex-ante}), where $\mathbf{q}^*$ is the optimal solution of the convex program \ref{eq:con-ex-ante}. Consider $\x{t}(v)=\tfrac{1}{2}\x{t}^*(v)$. We have:
\begin{equation*}
\x{t}(v)=\tfrac{1}{2}\x{t}^*(v)\overset{(1)}{\leq} \tfrac{1}{2}\overset{(2)}{\leq} 1-\tfrac{1}{2}\sum_{t'<t}q^*_{t'}\overset{(3)}{=}1-\Exv{v_{t'}\sim\Fi{t'}}{\x{t'}(v_{t'})}
\end{equation*}
where inequality~(1) holds as $\x{t}^*(v)\leq 1$, inequality~(2) holds as $\sum_{t'<t}q^*_{t'}\leq \sum_{t'}q^*_{t'}\leq 1$, and equality~(3) holds as $\Exv{v_{t'}\sim\Fi{t'}}{\x{t'}(v_{t'})}=\frac{1}{2}\Exv{v_{t'}\sim\Fi{t'}}{\x{t'}^*(v_{t'})}=\frac{1}{2}q^*_{t'}$. Therefore, $\{\x{t}(v)\}$ forms a feasible assignment for the \ref{eq:lp-opt-on}. By applying \Cref{prop:rounding-exact}, there exists a feasible randomized policy that implements $\{\x{t}(v)\}$; this policy obtains exactly the same allocation probabilities as $\{\x{t}(v)\}$ and obtains an expected value equal to the objective value of \ref{eq:lp-opt-on} for this assignment. Clearly, this objective value is at least $\frac{1}{2}$ of the expected value of the optimum offline, as the optimal value of the expected LP relaxation is an upper-bound on the expected maximum value. Finally, note that the randomized policy implementing $\x{t}(v)$, described in the proof of \Cref{prop:rounding-exact}, is exactly equivalent to Algorithm~\ref{alg:non-identical}.
\qed
\end{proof}


\begin{algorithm}[H]
\small
\caption{Online policy for non-identical distributions}
\begin{algorithmic}[1]
\State \textbf{input}: Distributions $\{\Fi{1},\ldots,\Fi{n}\}$
\State  Compute the optimal solution of (\ref{eq:con-ex-ante}). Let $\mathbf{q}^*$ be this optimal solution.
 $t\leftarrow 0$
\State \textbf{While} $\left(\left[~\textrm{item is not allocated}~\right]~\&~\left[ t\leq n \right]\right)$~\textbf{do}
 \State \quad\quad$t\leftarrow t+1$.
 \State \quad \quad Post a price $p_t\triangleq T_t(q^*_t)$.
 \State \quad \quad \textbf{If}~price $p_t$ gets accepted ($v_t\geq p_t$),~\textbf{then}
 \State  \quad\quad \quad\quad Allocate the item with probability $\displaystyle\frac{1}{2-\sum_{t'<t}q^*_{t'}}$.
\State \textbf{End}
\end{algorithmic}
\label{alg:non-identical}
\end{algorithm}
%
%
%
%
%
%
%
%
%
\paragraph{Rounding for the identical distributions.}
Can we round the optimal solution of the expected LP for the case of identical distributions, and obtain the improved bound of $1-\tfrac{1}{e}$, or even the optimal bound in \cite{correa2017posted}? Interestingly, by incorporating a careful rounding of the expected LP and using the LP of optimum online, we can obtain a mechanism which is posting the single price of $T(\tfrac{1}{n})$ and show that it achieves at least $1-\tfrac{1}{e}$ fraction of the optimum offline (hence an alternative proof for a similar result in \cite{correa2017posted} and \cite{esfandiari2017prophet}). 

\begin{proof}{Proof.}
Due to the symmetry, the optimal solution of \ref{eq:con-ex-ante} is attained at $q^*_i=\tfrac{1}{n}$, and hence $\x{t}^*(v)=\mathds{1}\{v\geq T(\tfrac{1}{n})\}$. Let $\gamma\triangleq 1-\tfrac{1}{n}$, and consider the solution $\x{t}(v)=\gamma^t\cdot \x{t}^*(v)$. Note that $\Exv{v\sim\Fi{}}{\x{t}^*(v)}=\tfrac{1}{n}=1-\gamma$ for all $t$. Moreover, we have $\x{t}(v)=\gamma^t\cdot \x{t}^*(v)\leq \gamma^t$, simply because $\x{t}^*(v)\leq 1$. Therefore, 
\begin{equation*}
1-\sum_{t'<t}\Exv{v_{t'}\sim\Fi{}}{\x{t'}(v_{t'})}=1-\sum_{t'<t}\gamma^{t'}\cdot\Exv{v_{t'}\sim\Fi{}}{\x{t'}^*(v_{t'})}=1-\frac{1}{n}\frac{1-\gamma^t}{1-\gamma}=\gamma^t
\end{equation*}
where in the last equality we used $\gamma=1-\tfrac{1}{n}$. So, $\x{t}(v)\leq 1-\sum_{t'<t}\Exv{v_{t'}\sim\Fi{}}{\x{t'}(v_{t'})}$, and hence forms a feasible solution to \ref{eq:lp-opt-on} for any $\pi$. \Cref{prop:rounding-exact} suggests that there exists a randomized policy that implements this feasible assignment. In fact, similar to the proof of \Cref{prop:rounding-exact}, the final policy should post the price $T(\tfrac{1}{n})$, and if $v\geq T(\tfrac{1}{n})$ should accept it with probability:
\begin{equation*}
\displaystyle\frac{\gamma^t}{1-\sum_{t'<t}\Exv{v_{t'}\sim\Fi{}}{\x{t'}(v_{t'})}}=\frac{\gamma^t}{1-\frac{1}{n}\sum_{t'<t}\gamma^{t'}}=\frac{\gamma^t}{1-\frac{1}{n}\frac{1-\gamma^t}{1-\gamma}}=1,
\end{equation*}
where the last equality again holds because $\gamma=1-\tfrac{1}{n}$. So, the exact rounding policy simply suggests posting the single price $T(\tfrac{1}{n})$. To compare the expected value of this policy with that of the optimum offline, we only need to compare the objective value of \ref{eq:lp-opt-on} at this solution with the objective value of expected LP, thanks to \Cref{prop:rounding-exact}. We have:
\begin{equation*}
\displaystyle \sum_{t=1}^{n}\Exv{v_{t}\sim\Fi{}}{v_{t}\cdot \x{t}(v_{t})}= \sum_{t=1}^{n}\gamma^t\cdot\Exv{v_{t}\sim\Fi{}}{v_{t}\cdot \mathds{1}\{v_t\geq T(\frac{1}{n})\}}=V(\frac{1}{n})\sum_{t=1}^n \gamma^t=V(\frac{1}{n})\frac{1-\gamma^{n+1}}{1-\gamma}~,
\end{equation*}
where the right-hand-side is equal to $n\cdot V(\tfrac{1}{n})\cdot (1-(1-\tfrac{1}{n})^{n+1})$ as $\gamma=1-\tfrac{1}{n}$. Finally, the optimal objective value of expected LP is equal to $n\cdot V(\tfrac{1}{n})$ and $(1-(1-\tfrac{1}{n})^{n+1})\geq 1-\frac{1}{e}$, which completes the proof.
\qed
\end{proof}

 \end{APPENDIX}
%
%





\bibliographystyle{informs2014}
\bibliography{refs}

\end{document}